\documentclass[onefignum,onetabnum]{siamonline190516}

\usepackage{lipsum}
\usepackage{amsfonts}
\usepackage{graphicx}
\usepackage{epstopdf}
\usepackage{algorithmic}
\usepackage{booktabs,subcaption}
\ifpdf
  \DeclareGraphicsExtensions{.eps,.pdf,.png,.jpg}
\else
  \DeclareGraphicsExtensions{.eps}
\fi

\usepackage{enumitem}
\setlist[enumerate]{leftmargin=.5in}
\setlist[itemize]{leftmargin=.5in}

\newsiamremark{remark}{Remark}
\newsiamremark{hypothesis}{Hypothesis}
\crefname{hypothesis}{Hypothesis}{Hypotheses}
\newsiamthm{claim}{Claim}

\headers{Network Interpolation}{T. Reeves, A. Damle, and A. R. Benson}

\title{Network Interpolation\thanks{Submitted to the editors on June 28, 2019.
    \funding{This research was supported by NSF Award DMS-1830274,
      ARO Award W911NF19-1-0057, and ARO MURI.}}}

\author{Thomas Reeves\thanks{Center for Applied Mathematics, Cornell University, Ithaca, NY 14853
  (\email{tr352@cornell.edu}).}
\and Anil Damle\thanks{Department of Computer Science, Cornell University, Ithaca, NY 14853
	(\email{damle@cornell.edu}, \email{arb@cs.cornell.edu}).}
\and Austin R. Benson\footnotemark[3]}

\usepackage{amsopn}

\ifpdf
\hypersetup{
  pdftitle={Network Interpolation},
  pdfauthor={T. Reeves, A. Damle, and A. R. Benson}
}
\fi

\newcommand{\rev}[1]{{#1}}

\begin{document}

\maketitle

\begin{abstract}
  Given a set of snapshots from a temporal network we develop, analyze,
  and experimentally validate a so-called network interpolation scheme. Our
  method allows us to build a plausible, albeit random, sequence of graphs that transition
  between any two given graphs. Importantly, our model is well characterized by a
  Markov chain, and we leverage this representation to analytically estimate the hitting time (to a predefined distance to the target graph) and
  long term behavior of our model. These observations also serve to provide interpretation and justification for a rate parameter in our model. Lastly, through a mix of synthetic and real-world data
  experiments we demonstrate that our model builds reasonable graph trajectories
  between snapshots, as measured through various graph statistics.
  In these experiments, we find that our interpolation scheme compares
  favorably to 
  common network growth models, such as preferential attachment and triadic closure.
\end{abstract}



\section{Introduction}

Dynamic networks for temporal interactions of complex systems
are a pervasive model throughout the sciences~\cite{Holme-2012-temporal}.
They are used to analyze, for example,
interactions in social networks~\cite{Fowler-2008-happiness,Leskovec-2005-graphs},
communication systems~\cite{Kossinets-2008-vector-clocks,Leskovec-2008-planetary},
digital currency transactions~\cite{Kondor-2014-bitcoin,Paranjape-2017-motifs}, and
protein-protein interactions~\cite{Han-2004-evidence,Komurov-2007-revealing}.
Often, these dynamic datasets are recorded as a sequence of ``snapshots''
(also called ``slices''~\cite{Mucha-2010-multislice} or ``layers''~\cite{Kivela-2014-multilayer}),
where the snapshot represents the network at a single point in time or an aggregation of data over a period of
time. A sequence of snapshots is often the fundamental type of data used to
derive methods for dynamic network analysis~\cite{Araujo-2014-Com2,Dunlavy-2011-temporal,Gorovits-2018-LARC,Henderson-2010-forensics}.

Even if dynamic interactions occur in real time, there are a number of reasons
why one may only have access to a sequence of snapshots. A principal reason is
that data may only be collected at regular intervals.
Specifically, such scenarios are common in survey data.
For example, sociological studies record social networks of groups at
different points in
time~\cite{VanDeBunt-1999-friendship,Freeman-1980-semi,Michell-1997-pecking},
and U.S.\ Census data such as the Survey of Income and Program Participation
(SIPP) records job movement by surveying households at regular intervals.
An online analog to offline surveys is Web crawling.  In this scenario,
sequences of snapshot networks are recorded from a sequence of crawls that
collect network data and are subsequently analyzed for their dynamic
structure~\cite{Backstrom-2006-group,Lewis-2008-tastes,Mislove-2008-growth}.
In other cases, temporal network data is aggregated upon public release.
This happens due to privacy concerns in biomedical
data~\cite{Ball-2007-DAWN,Benson-2018-simplicial} or because the interaction is
associated with a regularly scheduled event, such as coauthorship of a computer
science conference paper in a given
year~\cite{Leskovec-2005-graphs,Ley-2009-DBLP}.

Independent of the manner in which a sequence of snapshots is obtained, a natural
problem is inferring what happens in the network between the snapshots, when the underlying true data is not available. Such reverse engineering would enable better real-time data
analysis, localization of structural changes in the graph, and understanding of
social, financial, or biological dynamics. More generally, we would like to generate a plausible sequence of graphs that takes us from one snapshot to the next. Beyond enabling exploration of the network between snapshots, such a process is valuable in the development of synthetic data for streaming algorithms by providing test sets (sequences of graphs) anchored to the real data.


A natural first approach to this problem is to appeal to network growth models
such as preferential
attachment~\cite{Albert-1999-scaling,Krapivsky-2000-connectivity}, triadic
closure~\cite{Jackson07,Jin01}, mixture models~\cite{Overgoor-2018-choosing}, or
stochastic actor oriented models~\cite{Snijders-1996-SAOM}.  In this setup, one
would start with a snapshot and use a network growth model, perhaps with parameters determined by the structure of the next snapshot, as a guess at how the
network evolves. However, there is no guarantee that after an appropriate number of steps such growth models will be close in
structure to the next snapshot, and indeed we see that this is the case in our experiments outlined in \Cref{sec:experiments}. Across all of our experiments, the behavior we observe with
these models is that when initiated with one snapshot they do not have the same
structure of the network at the time of the next snapshot.


\begin{figure}
	\includegraphics[width=0.495\linewidth]{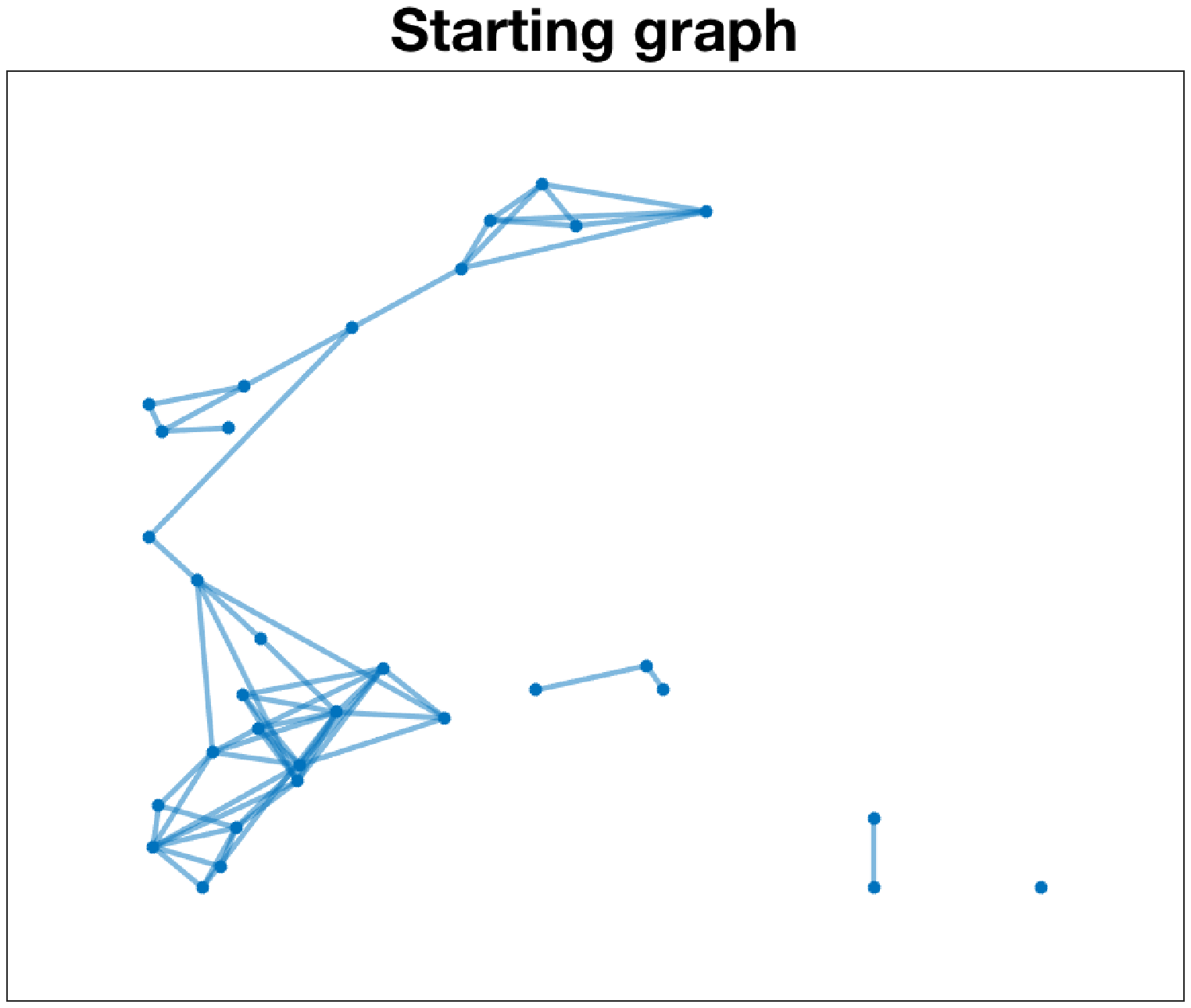}
	\includegraphics[width=0.495\linewidth]{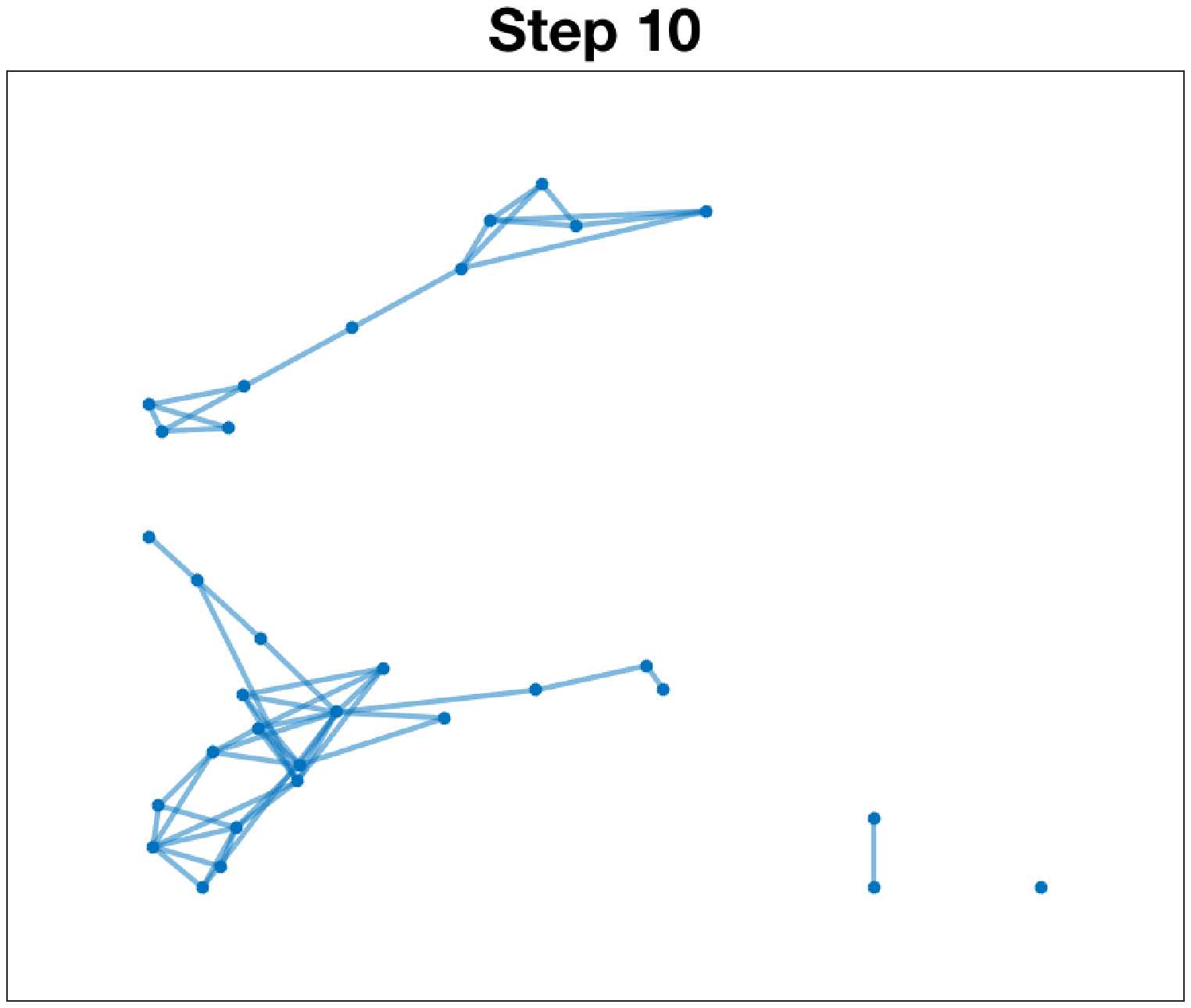} \\
	\vspace{-5mm}
	\includegraphics[width=0.495\linewidth]{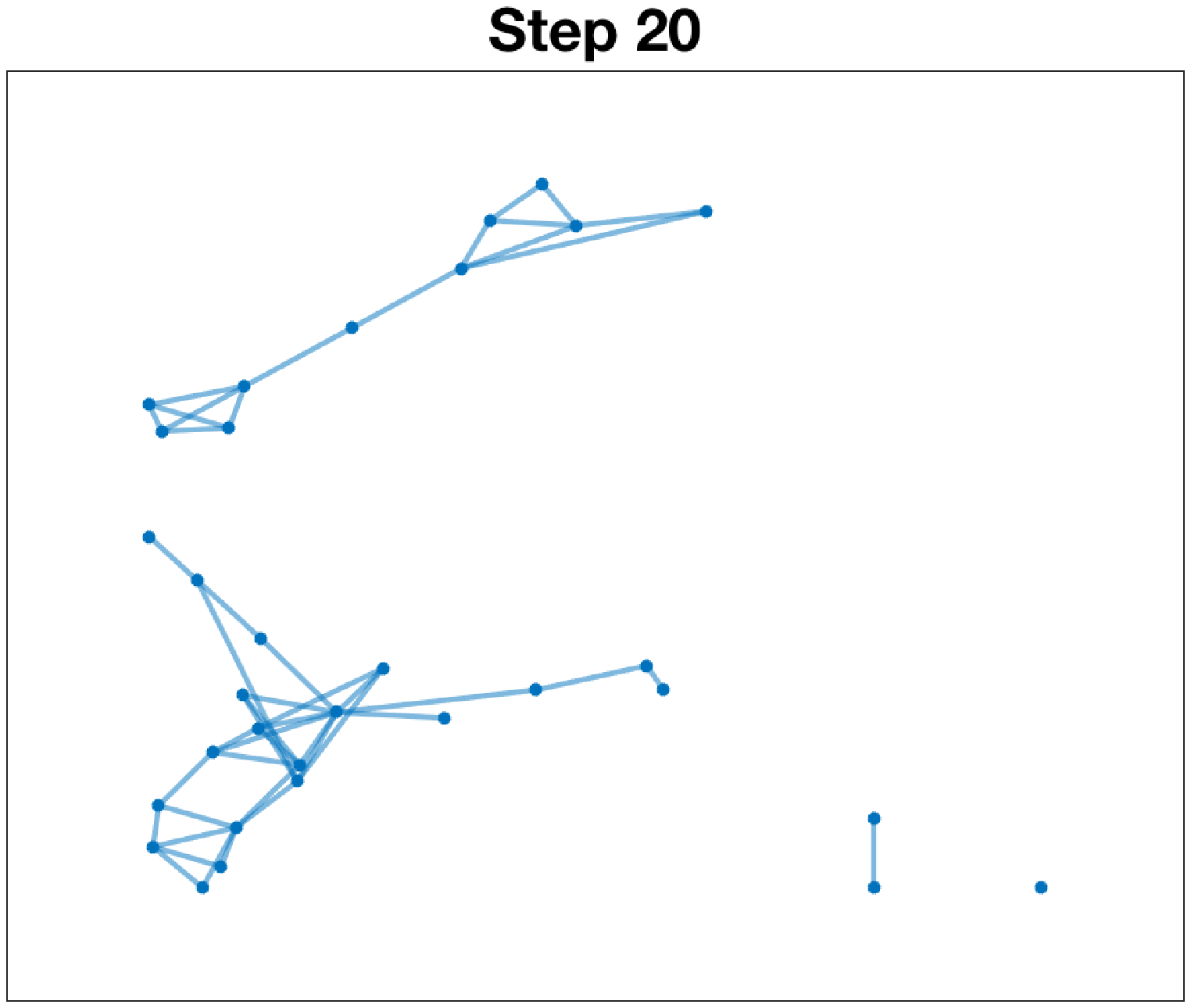}
	\includegraphics[width=0.495\linewidth]{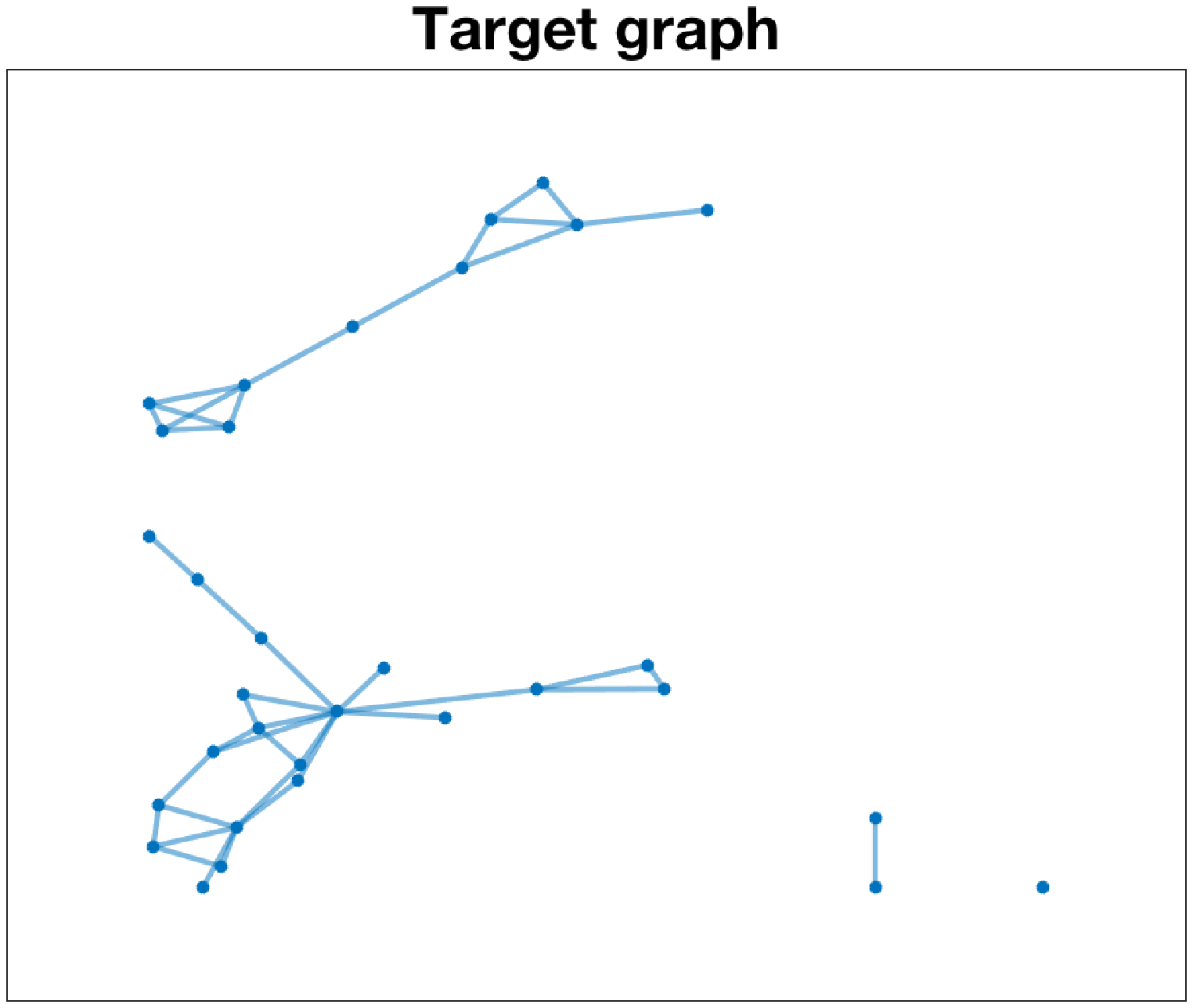}
	\vspace{-5mm}
	\caption{Example graphs from an instance of our network interpolation model. 
	The top left is the starting graph (Step 0), and the bottom right is the target graph 
	(which is reached at Step 31); these two graphs are snapshots from a dynamic 
	social network constructed from survey responses~\cite{VanDeBunt-1999-friendship}.
	 The other two graphs interpolate between the starting and target graphs and occur at the indicated 	
	 intermediate steps indicated, where a single edge is added or deleted in each step
	 (the target edit distance is set to $d_t = 0$ and the rate parameter is set to $s=1$).}
	\label{fig:vdbuntgraphs}
\end{figure}

Colloquially, employing a growth model looks like using an extrapolation method to try and move from one snapshot to the next\textemdash the only considerations are the choice of the parameters in a growth model and the starting snapshot. In contrast, because we have a sequence of snapshots a more apt analogy is to try and build an interpolation strategy. Here we will use both the starting and ending snapshots and explicitly generate a path between them. To illustrate the subsequent discussion, \Cref{fig:vdbuntgraphs} shows four graphs produced in an instance of our network interpolation model. 

We propose a temporal network model that, given two sequential snapshots,
provides a path \emph{of graphs} that begins at one snapshot and ends exactly at the
next snapshot\textemdash i.e., it exactly interpolates the snapshots. More specifically, the model provides a feasible sequence of
additions and deletions of edges that transitions between the two snapshots. We stress that a key benefit of our model is that it does \emph{not} simply interpolate one-dimensional statistics of the snapshots (such as the clustering coefficient). It actually provides a sequence of graphs whose statistics interpolate those of the snapshots.

Of course, absent sufficient additional side information, there could be an infinite
number of potential sequences from one graph to the next. Thus, our model is determined by the process through which we build the interpolation. More specifically, our model is based on a Markov chain on
graphs. Given a starting graph $G$ and a target graph $H$, our model makes a
sequence of edits to the graph $G$; at each step, a biased coin determines
whether or not we increase or decrease the edit distance between $H$ and $G$ by
one, where the bias depends on the current edit distance.

While our model is a Markov chain on the space of all graphs of appropriate
size, we show how to analyze its structure using a much simpler Markov chain on the set
of edit distances between the starting and target graphs. This reduces the
dimension of the chain from exponential in the number of nodes to quadratic in
the size of the starting and target graph. By assuming a flexible 
bias parameter in the type of transition (increasing or decreasing
edit distance) we are able to provide an analytic expression for the expected hitting time
of the Markov chain, thereby characterizing the expected number of edge additions and deletions
needed to reach the target graph as a function of a rate parameter. We also theoretically show that the limiting
distribution is concentrated around the target graph.  Our theory is also
validated through numerical experiments.

A further benefit of our model is that it may be adapted beyond just the
interpolation setting. Specifically, rather than being asked to explicitly hit
the target graph it can also evolve to a graph within a given edit
distance of the target graph and then fluctuate around that point. This is
useful in several settings. For example, one may assume a given snapshot is
noisy and matching it exactly may not be desirable. Alternatively, this process
may be used to construct long term trajectories of graphs that fluctuate around
a given target. These trajectories, anchored by real-world graphs, may then be
used in the development and testing of streaming algorithms.

Experimentally, we use our model to analyze a number of synthetic and real-world
network snapshot sequences. We show that our model can interpolate between
snapshots, while common extrapolatory growth models deviate substantially from
the exhibited structure even when we fit growth parameters to the snapshots. One key example is that our
network interpolation scheme often maintains clustering between snapshots (as measured by the mean and global clustering coefficients),
whereas growth models (even ones based on triadic closure) fail to reasonably
represent the clustering between snapshots in communication, social, and
collaboration networks.

\section{Related work}

The most relevant models in the literature are snapshot models
\cite{Ghasemian16,Xu14}, which consist of probabilistically generated graphs
whose parameters (such as edge connection probabilities and vertex labels) vary
randomly over a series of discrete time steps. Typically, the number of snapshots
is very small compared to the size of the entire graph, whereas the granularity of our method
is at a per-edge level. The main goal of existing snapshot studies is
statistical inference of the model parameters.
For example, Bhattacharjee \emph{et al.} consider a series of
independently drawn graphs with a given community structure that exhibits a
sudden change and study the statistical estimation of the time
of the change~\cite{Bhattacharjee18}.
We perform a similar change-point detection experiment in \Cref{sec:synthetic}.

Importantly, snapshot models do not capture the fine-grain changes that take place
in a dynamic network. Network data is becoming increasingly available through fine-grain measurements of temporal
interactions~\cite{Farajtabar-2015-coevolve,Holme-2012-temporal,liu2018sampling,Scholtes-2017-networks},
and snapshot graphs are designed to capture coarse-grain structure.
Also, snapshot models do not yet fully
harness temporal structure to identify the emergence of structure in an
\emph{evolutionary} framework. Aldecoa and Mar\'{i}n head toward this
direction by proposing a rewiring process from one graph to another as a
benchmark for evaluating community detection algorithms~\cite{Aldecoa13}.
In order to tackle the problem of understanding structure, we need fine-grain models for network evolution.

Another relevant model is the stochastic actor oriented model~\cite{Snijders-1996-SAOM}. 
Nodes are conceptualized as actors who control their ties to the other actors in the
network. Each actor can change ties at time points determined by a Poisson rate
function, and 
the probability of the actor toggling a tie given snapshots is estimated \cite{VanDeBunt-1999-friendship}. By
contrast, our model processes edges one by one, does not characterize the
individual nodes in the network, and operates on a global level.

Finally, another related idea to our research is so-called ``network archaeology,'' which
reconstructs network histories given a present day
snapshot~\cite{Navlakha11,Young2018}. These methods infer growth model parameters by
running growth models in reverse to find likely ancestors of a graph. Our model instead
interpolates from one snapshot to the next.

\section{Network interpolation model}

We now concretely present and analyze our model for network interpolation based on given starting and target graphs. 
All code and data used for this paper are available at \url{https://github.com/tr-maker/networkinterpolation}.

\subsection{Model description}
We initialize our model with the starting graph $G$ and target graph $H$. For example, we could be given two snapshots of a social network and want to create a sequence of graphs interpolating between them. At each step, our model makes an update to $G$ so that it represents the current graph. (The graph $H$ remains fixed throughout the interpolation.) 
Our graph evolution model is in terms of the \emph{graph edit distance} $d$ between $G$ and $H$, which is the minimum number of \emph{moves} (edge or vertex additions or deletions) needed to go from the former graph to the latter. Without loss of generality, we assume that $G$ and $H$ have the same number of vertices $n$,\footnote{This is not particularly limiting as there is no restriction on $G$ or $H$ having multiple connected components or nodes with no edges.} and that moves consist only of edge additions and deletions. At each step, we change $G$ so that 
\begin{itemize}
	\item with probability $\varphi(d)$, we make some \emph{advancing} move that decreases the edit distance $d$ by $1$, or
	\item with probability $1-\varphi(d)$, we make some \emph{regressing} move that increases the edit distance $d$ by $1$.
\end{itemize}
We have freedom in how the advancing and regressing moves are chosen.\footnote{Because we are only modeling changes in the graph we omit the
	possibility for the graph to remain the same. This removes any notion of time
	from our model, though it could be introduced by further modeling when the
	edits occur.} 
Advancing moves (1) add an edge that is in $H$ but
not $G$ or (2) delete an edge that is in $G$ but not in $H$. 
Regressing moves (1) delete an edge in $G$ that is also in $H$ or (2) add an edge
to $G$ that is not in $H$. The simplest model to analyze is when all these
possibilities are allowed, and each advancing (resp.\ regressing) move is chosen
uniformly at random from the set of all possible advancing (resp.\ regressing)
moves. However, we may for example wish to disallow adding an edge to $G$ that
is not in $H$. In such a model, we may have \emph{outdated edges} (edges from
the initial graph that are not in $H$) but never \emph{false edges} (edges not in the initial graph nor $H$).

Another parameter we can control is how the advancing probability $\varphi$ varies as a function of the current edit distance $d$ to the target. 
Here we choose it as a sigmoid function of $d$ that is centered at a given distance $d_t$. In the long term this process generates graphs that fluctuate around graphs that are edit distance $d_t$
from the target graph, and we call $d_t$ the \emph{target edit distance}.
We also need to deal with boundary conditions of $d=0$ and
$d=d_m$, where $d_m = \binom{n}{2}$ is the maximum possible edit distance. More formally, $\varphi$ is
specified as follows:
\begin{itemize}
	\item $\varphi$ is a monotonically increasing function of $d$.
	\item $\varphi(0) = 0$, $\varphi(d_m) = 1$, and $\varphi(d_{t}) = \frac{1}{2}$.
\end{itemize}
These properties are realized if we define, for $0 < d < d_m$,
\begin{equation}
\label{eq:varphi}
\varphi(d) = f\left(\frac{d-d_t}{s}\right),
\end{equation}
where $f$ is a sigmoid function and $s$ is a rate parameter that controls the hitting time to $d_t$ and the spread of distances around $d_t$ in the time-averaged limiting distribution.
The standard logistic function $f(x) = (1+e^{-x})^{-1}$ makes our model analytically tractable,
and we use this function throughout the paper.

Once we have chosen $\varphi$, \Cref{alg:networkinterp} compactly summarizes our model when we choose to run it for a fixed number of steps. (It is easily adaptable to settings where, \emph{e.g.,} we wish to terminate when the target distance is first reached\textemdash setting this target distance parameter to zero yields an interpolation scheme.) We note that if we want to store the dynamic graph at each step, we can save space by storing a sequence of edge updates to the initial graph. In addition, we do not need to store the set of regressing moves explicitly since
they can be inferred from the advancing moves. 

\begin{algorithm}
	\caption{Network interpolation}
	\label{alg:networkinterp}
	\begin{algorithmic}[1]
	    \STATE Input: $G =$ starting graph, $H =$ target graph, $T =$ number of steps, $s =$ rate, $d_t =$ target edit distance, $\varphi(\cdot,s,d_t) =$ advancing probability
		\STATE initialize $\mathcal{X} = \left\{\text{advancing moves from $G$ to $H$}\right\}$ 
		\STATE initialize $\mathcal{Y} = \left\{\text{regressing moves from $G$ to $H$}\right\}$
		\STATE $d \gets \text{edit distance between $G$ and $H$}$
		\FOR {$\tau = 1:T$}
		\STATE $\text{bool} \sim \text{Bernoulli}(\varphi(d))$
		\IF {bool}
			\STATE // Make an advancing move.
			\STATE sample a move from $\mathcal{X}$ and update $G,$ $\mathcal{X},$ and $\mathcal{Y}$ 
			\STATE $d \gets d-1$
		\ELSE 
			\STATE // Make a regressing move.
			\STATE sample a move from $\mathcal{Y}$ and update $G,$ $\mathcal{X},$ and $\mathcal{Y}$
			\STATE $d \gets d+1$
		\ENDIF
		\ENDFOR
	\end{algorithmic}
\end{algorithm}

\Cref{alg:networkinterpimp} describes in more detail an efficient implementation of \Cref{alg:networkinterp} \rev{for undirected graphs} given the adjacency matrices of the starting and target graphs. If $A$ is the strictly upper triangular part of the adjacency matrix of the starting graph and $B$ the strictly upper triangular part of the adjacency matrix of the target graph, then the matrix $U = B-A$ keeps track of the advancing and regressing moves.\footnote{\rev{If instead the starting and target graphs are directed, then we let $A$ be the entire adjacency matrix of the starting graph, $B$ be the entire adjacency matrix of the target graph, and $U = B-A$.}}
An entry of $1$ in $U$ indicates an edge that is \rev{not in} the current graph but in the target graph;
an entry of $-1$ in $U$ indicates an edge that is in the current graph but not in the target graph;
and an entry of $0$ in $U$ indicates an edge that is both in the current graph and the
target graph, or not in either. Therefore, the number of nonzeros in $U$ is the
number of advancing moves (or the edit distance $d$ between the current graph and the target graph), and the number of zeros of $U$ is the number of regressing moves. 

\begin{algorithm}
	\caption{Network interpolation (sketch of efficient implementation \rev{for undirected graphs})}
	\label{alg:networkinterpimp}
	\begin{algorithmic}[1]
		\STATE Input: $A =$ strictly upper triangular part of adjacency matrix of starting \rev{undirected} graph, $B =$ strictly upper triangular part of adjacency matrix of target \rev{undirected} graph, $T =$ number of steps, $s =$ rate, $d_t =$ target edit distance, $\varphi(\cdot,s,d_t) =$ advancing probability
		\STATE initialize $U = \text{$B-A$}$ 
		\STATE $d \gets \text{number of nonzeros of $U$}$
		\FOR {$\tau = 1:T$}
		\STATE $\text{bool} \sim \text{Bernoulli}(\varphi(d))$
		\IF {bool}
		\STATE // Make an advancing move.
		\STATE sample an index $i$ from the nonzero entries of $U$
		\IF {$U(i) = 1$}
		\STATE $A(i) \gets 1$  \;\;// Add the corresponding edge to $A$.
		\ELSE
		\STATE $A(i) \gets 0$ \;\;// $U(i) = -1$, so we delete the corresponding edge from $A$.
		\ENDIF
		\STATE $U(i) \gets 0$
		\STATE $d \gets d-1$
		\ELSE 
		\STATE // Make a regressing move.
		\STATE sample an index $i$ from the zero entries of $U$
		\IF {$B(i) = 1$}
		\STATE $A(i) \gets 0$,\; $U(i) \gets 1$ \;\;// Delete the corresponding edge from $A$.
		\ELSE
		\STATE $A(i) \gets 1$, $U(i) \gets -1$ \;\;// $B(i) = 0$, so we add the corresponding edge to $A$.
		\ENDIF
		\STATE $d \gets d+1$
		\ENDIF
		\ENDFOR
	\end{algorithmic}
\end{algorithm}

\rev{Our algorithm scales well with respect to space and time complexity. Before giving a more formal analysis, we first give a sense of the practical running time based on one of our larger experiments in \Cref{sec:experiments}. There, we interpolate between 9 snapshots of a coauthorship graph with several hundred thousand nodes and edges in each snapshot. We set $s=1$ as the rate parameter and $d_t = 0$ as the target edit distance, and we sequentially ran our algorithm by using consecutive snapshots as the respective starting and target snapshots until the target snapshot was reached. The entire sequence of 8 consecutive interpolations covered 2862784 steps and took 49.2 seconds in total, or 17.2 microseconds per interpolation step, on an early 2015 MacBook Pro with 2.7 GHz and 8 GB of RAM.
(Our implementation is not highly optimized, and one possible approach for improving performance in practice 
would be to make the advancing and regressing moves in batches.)
}

\rev{We now analyze the space and time complexity of our algorithm theoretically. If $A$ and $B$ are sparse, then $U$ is sparse since it contains at most the combined number of nonzeros of $A$ and $B$. As the algorithm progresses, $U$ becomes sparser with each advancing move and denser with each regressing move. Due to our assumptions on the advancing probability $\varphi(\cdot,s,d_t)$, with high probability the sparsity of $U$ will be on the order of the sparsity of $A$ and $B$ throughout the algorithm. In other words, the total storage required for the algorithm is on the order of the storage required to store the starting and target graphs.}

\rev{As for the running time of the algorithm, initializing $U$ takes time linear in the number of nonzeros of $A$ and $B$. After that, each step of the algorithm depends on the complexity of the following three operations: 1) sampling a random index from the nonzero entries of $U$ (for an advancing move), 2) sampling a random index from the zero entries of $U$ (for a regressing move), and 3) updating an entry of $U$ (and similarly, updating an entry of $A$). In the case of uniform sampling and sparse $U$, the three operations can be made to take amortized constant time. By storing the entries of $U$ in a dynamically growing array and a hash map, where the entries of the array consist of the signed edges ($1$ or $-1$) in $U$ and the hash map maps each signed edge to its index in the array, the operations of insertion, deletion, search, and sampling a uniformly random element can be performed in amortized constant time. In particular, this takes care of operations 1) and 3). To implement operation 2), we can use rejection sampling, drawing uniformly random edges from the set of all possible edges until we draw an edge whose corresponding entry in $U$ is zero. If the number of nonzeros in $U$ is linear in the number of vertices, then this process takes a constant number of iterations with high probability. Altogether, each step of the algorithm takes amortized constant time. Therefore, the total running time is on the order of the number of nonzeros of $A$ and $B$, plus the order of the number of steps $T$.}

\Cref{fig:editdistances} shows how various choices of $s$ in \cref{eq:varphi} affect the series of edit distances from the target graph. 
For \Cref{fig:editdistances} and the remainder of the figures analyzed in \Cref{sec:analysis}, 
the starting graph is a 50-node Erd\H{o}s-R\'{e}nyi random
graph with edge connection probability 0.5, the target graph is a 50-node
stochastic block model divided into two equally sized clusters with in-cluster
edge connection probability 0.9 and out-of-cluster edge connection probability
0.1, and the target edit distance $d_t$ is equal to $10$.

\begin{figure}
	\centering
	\includegraphics[width=0.5\columnwidth]{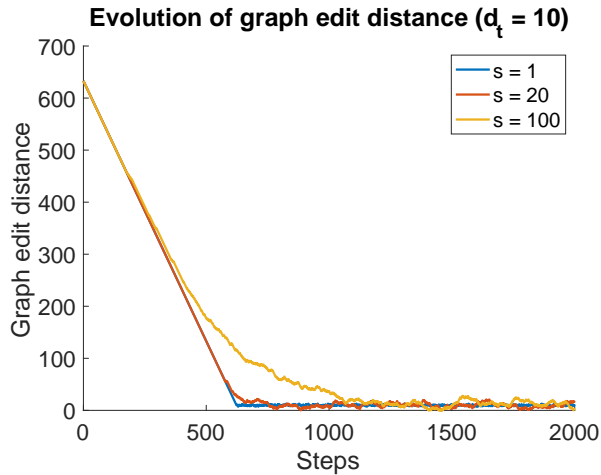}
	\caption{Evolution of edit distance from an  Erd\H{o}s-R\'{e}nyi random starting graph to a target graph of 2 communities,
	where transitions follow \cref{eq:varphi}. The higher the rate parameter $s$, the longer the hitting time to $d_t = 10$ and the 
	more spread out the distribution of edit distances.}
	\label{fig:editdistances}
\end{figure}

\subsection{Model analysis}\label{sec:analysis}

From \Cref{alg:networkinterp}, our model is a Markov chain on graphs. However, many interesting properties of our
model can be gleaned by looking at just the edit distances at each step of the
interpolation.  Our model naturally describes a stochastic process on edit distances 
$\{0, \ldots, d_m\}$; this process is in fact Markovian
\emph{provided that an advancing move and a regressing move are always possible for $0 < d <  d_m$}.
This very mild assumption holds exactly in the simplest case where all
possible advancing and regressing moves are allowed,\footnote{It is not strictly
	true in the case that we disallow false edges (i.e., disallow adding an edge
	to the current graph $G$ that is not in the target graph $H$), since every
	time an edge that is not in $H$ is removed from $G$, the edit distance for
	which a regressing move is impossible decreases by one.}  and our theoretical
estimates under this assumption closely match the experiments.


With the Markov assumption, $\varphi(i)$ describes the transition probability from edit distance $i$ to edit distance $i-1$. The transition probability matrix therefore has $1 - \varphi(i)$ above the diagonal and $\varphi(i)$ below the diagonal:
\begin{align}
\label{1dtransition}
\begin{bmatrix}
& 1\\
\varphi(1) & & 1-\varphi(1) \\
& \varphi(2) & & 1-\varphi(2) \\
& & \ddots & & \ddots \\
& & & \varphi(d_m-1) & & 1-\varphi(d_m-1). \\
& & & & 1
\end{bmatrix}.
\end{align}

With this, we analyze two important properties of our model: the limiting distribution of edit distances from the target graph (\Cref{thm:limitdist}) and the hitting time to a distance $d_t$ from the target graph (\Cref{thm:hittingtimes}). \rev{For these two quantities, we give formulae that can be computed explicitly in terms of the rate parameter $s$. These propositions are therefore of practical importance for choosing the rate parameter in modeling.} 
We then consider the properties of a random walk on the space of all graphs on $n$ vertices (\Cref{thm:markovongraphs}).
For these theoretical results, we assume that $\varphi$ is of the form in \cref{eq:varphi} with $f$ the standard logistic function. 
To experimentally validate the propositions, we use the simplest version of the model where each advancing (resp.\ regressing) 
move is chosen uniformly at random from the set of all possible advancing (resp.\ regressing) moves.

\subsubsection{Time-averaged limiting distribution of edit distances}
The Markov chain as described has period 2: the state at any given time step is either an even or an odd distance away from the initial edit distance. Nevertheless, the Markov chain has a time-averaged limiting distribution in the sense that for each state, the proportion of time spent in that state before step $n$ converges as $n\to\infty$ \cite{KS}. 

The left eigenvector corresponding to the eigenvalue 1 of the transition probability matrix is represented in unnormalized form as follows:
\begin{align}
\begin{bmatrix}
\varphi(1)\cdots\varphi(d_m-1) \\
\varphi(2)\cdots\varphi(d_m-1) \\
(1-\varphi(1))\varphi(3)\cdots\varphi(d_m-1) \\
(1-\varphi(1))(1-\varphi(2))\varphi(4)\cdots\varphi(d_m-1) \\
\cdots \\
(1-\varphi(1))\cdots(1-\varphi(d_m-2)) \\
(1-\varphi(1))\cdots(1-\varphi(d_m-1))
\end{bmatrix}^T.
\end{align}
The time-averaged limiting distribution $v = \{v_i\}_{i=0}^{d_m}$ of the Markov chain is this eigenvector normalized by the sum of the entries. Although it is possible to do this eigenvector computation directly, we have the following simpler formula for \rev{approximately computing} entries of the eigenvector around $d_t.$
\begin{proposition}
	\label{thm:limitdist}
	Assume that $d_m \ge 2d_t$. There exists a constant $C$ such that for every $d_t \ge 2$ and $0 < \epsilon < 1$, if $s < C\frac{d_t^2}{\log 1/\epsilon}$, each component of the time-averaged limiting distribution satisfies the following approximation for $0 < k < d_t$:
	\begin{equation}
	\label{eq:limitdist}
	\left\lvert v_{d_t \pm k} 
	- \frac{e^{-\frac{k(k-1)}{2s}} + e^{-\frac{k(k+1)}{2s}}}{2\left(e^{-\frac{(d_t-1)d_t}{2s}} + 2\sum_{i=0}^{d_t-2} e^{-\frac{i(i+1)}{2s}}\right)} \right\rvert < \epsilon.
	\end{equation}
	
\end{proposition}
\begin{proof}
	We can write the entries of $v$ recursively:
	\begin{equation}
	\label{eq:vrecursive}
	v_{i+1}= \frac{1 - \varphi(i)}{\varphi(i+1)}v_{i},\quad 0 \le i < d_m.
	\end{equation}
	From \cref{eq:vrecursive}, we make two observations. First, if $\varphi(d_t + i) = 1 - \varphi(d_t - i)$ for all $i$, then
	\begin{equation}
	\label{eq:vproperty1}
	v_{d_t + i} = v_{d_t - i}.
	\end{equation}
	Second, if we iterate \cref{eq:vrecursive} $k$ times, we get
	\begin{align*}
		v_{i + k} = \frac{1-\varphi(i)}{\varphi(i + k)} \left(\frac{1}{\varphi(i+1)} - 1\right) \cdots \left(\frac{1}{\varphi(i+k-1)} - 1\right) v_i,
	\end{align*}
	so in particular
	\begin{equation}
	\label{eq:vproperty2}
	v_{d_t + i} = \frac{1/2}{f(\frac{i}{s})} \left(\frac{1}{f(\frac{1}{s})} - 1\right) \cdots  \left(\frac{1}{f(\frac{i-1}{s})} - 1\right) v_{d_t}.
	\end{equation}
	Now assuming that $f$ is the standard logistic function, combining \cref{eq:vproperty1} and \cref{eq:vproperty2} gives that
	\begin{align}
		v_{d_t \pm i} &= \frac{1}{2}(1 + e^{-\frac{i}{s}})e^{-\frac{1}{s}} \cdots e^{-\frac{i-1}{s}} v_{d_t}
		= \frac{1}{2}(1 + e^{-\frac{i}{s}})e^{-\frac{i(i-1)}{2s}} v_{d_t} 
		= \frac{e^{-\frac{i(i-1)}{2s}} + e^{-\frac{i(i+1)}{2s}}}{2} v_{d_t}. \label{eq:vdt+i}
	\end{align}
	If $d_m \ge 2d_t$, $d_t \ge 2$, and $s$ is small enough relative to $d_t$, we can approximate the time-averaged limiting distribution by assuming that $v_i$ is supported on $0 < i < 2d_{t}$. Then \cref{eq:vdt+i} gives
	\begin{equation}
		1 = v_{d_t} + 2\sum_{i=1}^{d_t-1} \frac{e^{-\frac{i(i-1)}{2s}} + e^{-\frac{i(i+1)}{2s}}}{2} v_{d_t} + \eta v_{d_t},
	\end{equation}
	where the error
	\begin{equation}
	\label{eq:etabound}
	\eta < c e^{-\frac{d_t(d_t-1)}{2s}}
	\end{equation}
	for some universal constant $c$. So
	\begin{align*}
		\frac{1}{v_{d_t}} = 
		1 + 2\left( \frac{1}{2} + \frac{1}{2}e^{-\frac{(d_t-1)d_t}{2s}} + \sum_{i=1}^{d_t-2} e^{-\frac{i(i+1)}{2s}} \right) + \eta
		= e^{-\frac{(d_t-1)d_t}{2s}} + 2\sum_{i=0}^{d_t-2} e^{-\frac{i(i+1)}{2s}} + \eta
	\end{align*}
	and so $v_{d_t} = \frac{1}{e^{-\frac{(d_t-1)d_t}{2s}} + 2\sum_{i=0}^{d_t-2} e^{-\frac{i(i+1)}{2s}} + \eta}$.
	Going back to \cref{eq:vdt+i}, we get
	\begin{equation}
	v_{d_t \pm k} = \frac{e^{-\frac{k(k-1)}{2s}} + e^{-\frac{k(k+1)}{2s}}}{2\left(e^{-\frac{(d_t-1)d_t}{2s}} + 2\sum_{i=0}^{d_t-2} e^{-\frac{i(i+1)}{2s}}\right) + 2\eta}.
	\end{equation}
	The parenthesized term in the denominator is greater than 1, independent of $d_t$ and $s$. So for \eqref{eq:limitdist} to hold, it suffices that $\eta$ is at most some constant times $\epsilon$, where the constant is also independent of $d_t$ and $s$. The condition for $s$ in the proposition follows from \eqref{eq:etabound}.
\end{proof}

\Cref{fig:editdistanceshist} shows that the analytic formula reasonably approximates the limiting distribution of edit distances if the distribution of edit distances is not too spread out so that a symmetric approximation is reasonable.

\begin{figure}
	\includegraphics[width=0.495\columnwidth]{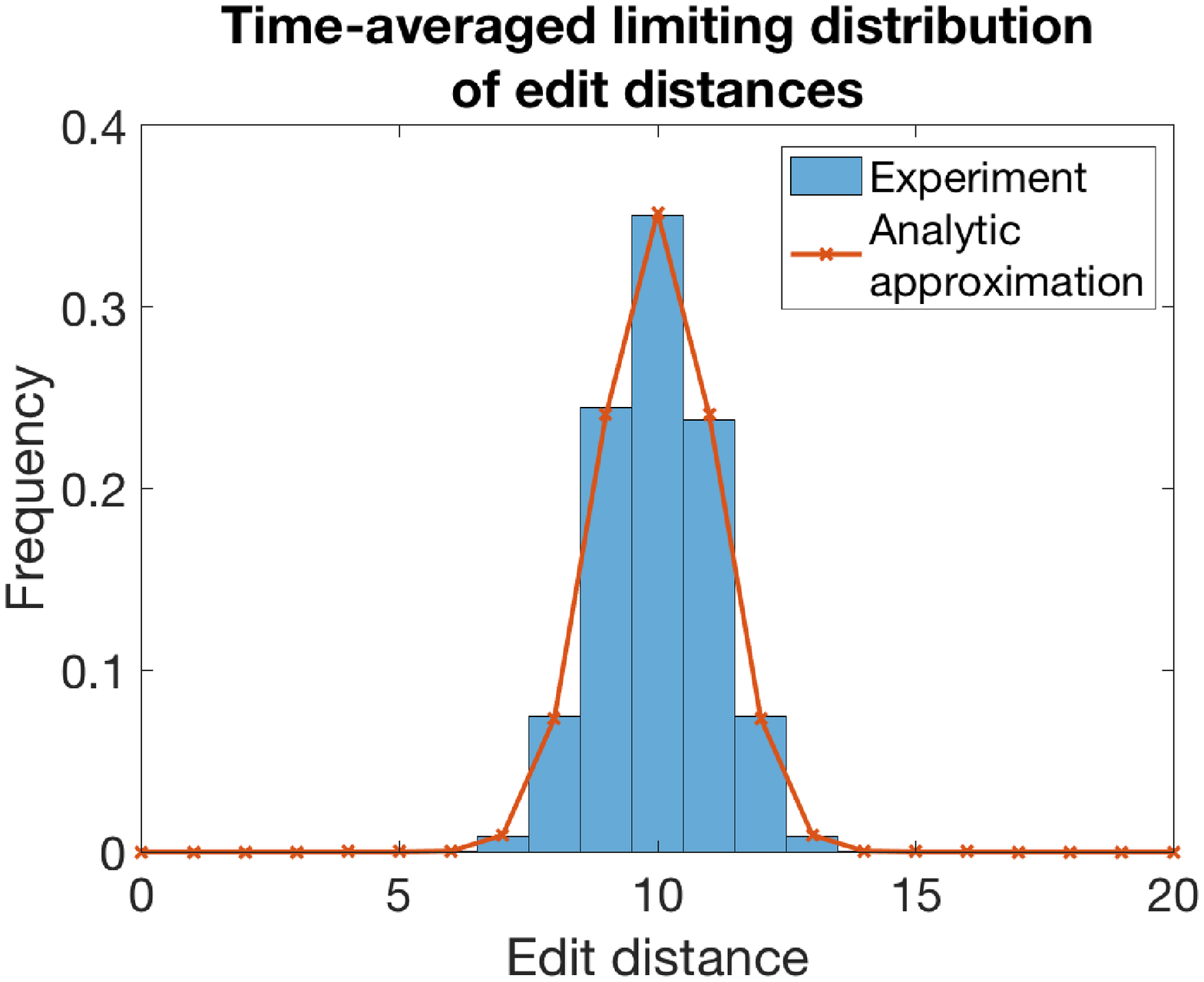} 
	\includegraphics[width=0.495\columnwidth]{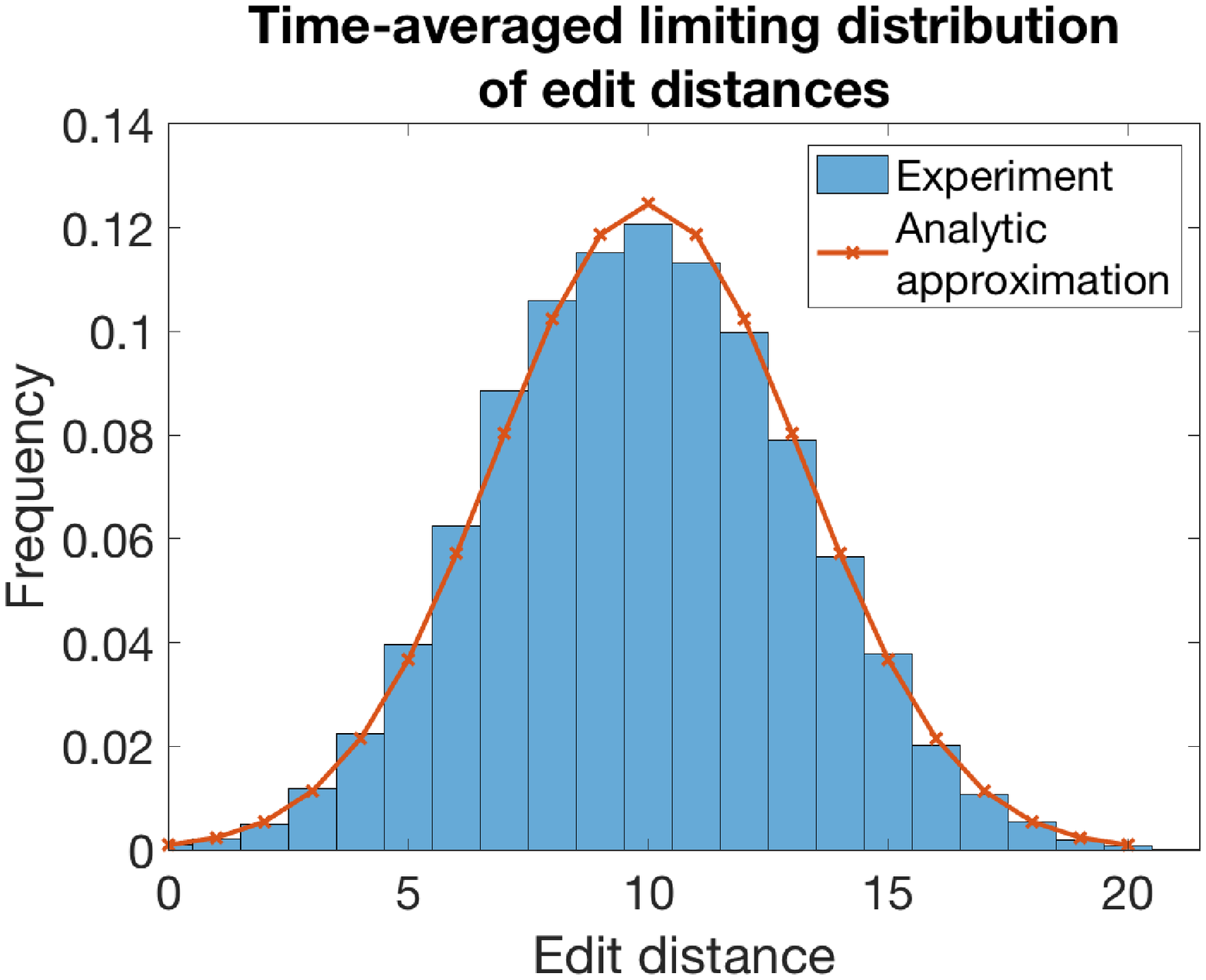}
	\caption{Histogram of the time-averaged limiting distribution of edit
		distances of an evolution of a random graph to a target graph
		consisting of 2 communities, together with its analytic approximation
		(\cref{thm:limitdist}) for two rate parameters ($s = 1$, left; and $s = 10$, right).
		The histograms are data from steps 10001--20000 of a simulation. The target distance $d_t = 10$.} 
	\label{fig:editdistanceshist}
\end{figure}


\subsubsection{Expected hitting time to the target edit distance}
The next proposition gives an analytic formula for the expected hitting time.
\begin{proposition}
	The expected hitting time $h_{d_o}$ to the target edit distance $d_t$, starting at the
	initial edit distance $d_o$ to the target graph, is
	\label{thm:hittingtimes}
	\begin{equation}
	h_{d_o} = (d_o-d_t) + 2 \left( \sum_{k=1}^{d_m - d_o} e^{-\frac{k(k+1)}{2s}} \frac{1 - e^{-\frac{k(d_o-d_t)}{s}}}{1 - e^{-\frac{k}{s}}} \right). \label{eq:sum}
	\end{equation}
	In particular, the hitting time is monotonic in $s$.
\end{proposition}
\begin{proof}
	Let $h_i$ be the expected time to hit $d_t$ starting from edit distance $d_i$. Then we have the linear recurrence
	\begin{equation}
	\label{eq:linearrecurrence}
	h_{d_t} = 0, \quad
	h_i = 1 + \varphi(i)h_{i-1} + (1-\varphi(i))h_{i+1}.
	\end{equation}
	Let $d_o > d_t$ be the starting edit distance. We are interested in $h_{d_o}$.
	From \cref{eq:linearrecurrence}, adding $\varphi(i)h_i - h_i$ to both sides, subtracting $\varphi(i)h_{i-1}$ from both sides, and dividing by $\varphi(i)$ gives
	\begin{equation}
	h_{i} - h_{i-1} = \frac{1}{\varphi(i)} + \left( \frac{1}{\varphi(i)} - 1 \right)(h_{i+1} - h_{i}),
	\end{equation}
	which is a recurrence in terms of the successive differences of hitting times.
	When $\varphi(d) = f\left( \frac{d-d_t}{s} \right)$ is the standard logistic function, then the above equation simplifies to
	\begin{equation*}
		h_i - h_{i-1} = 1 + e^{-\frac{i-d_t}{s}} + e^{-\frac{i-d_t}{s}}(h_{i+1}-h_i)
	\end{equation*}
	with initial condition $h_{d_m} - h_{d_m-1} = 1$. Rewriting, we get
	\begin{align}
		h_k - h_{k-1} = 
		1 + 2\left( e^{-\frac{(d_m-1)-d_t}{s}-\cdots-\frac{k-d_t}{s}} + 
		e^{-\frac{(d_m-2)-d_t}{s}-\cdots-\frac{k-d_t}{s}} + \cdots + e^{-\frac{k-d_t}{s}} \right). \label{eq:hkhk1}
	\end{align}	
	Now sum \cref{eq:hkhk1} from $k = d_t+1$ to $d_m$.
	The left-hand side is then the telescoping sum $h_{d_m} - h_{d_t} = h_{d_m}$, and we obtain
	\begin{align}
		h_{d_m} = \sum_{k = d_t + 1}^{d_m} \left( 1 + 2\left( e^{-\frac{d_m-1-d_t}{s}-\cdots-\frac{k-d_t}{s}}
		+ e^{-\frac{d_m-2-d_t}{s}-\cdots-\frac{k-d_t}{s}} + \cdots + e^{-\frac{k-d_t}{s}} \right) \right). \label{eq:hdm}
	\end{align}	
	To find the value of $h_{d_o}$, we sum \cref{eq:hkhk1} from $k = d_o+1$ to $d_m$.
	The left-hand side is the telescoping sum $h_{d_m} - h_{d_o}$. Plugging in \cref{eq:hdm} for the value of $h_{d_m}$ gives
	\begin{align*}
		h_{d_o} &= \sum_{k = d_t + 1}^{d_m} \left( 1 + 2\left( e^{-\frac{d_m-1-d_t}{s}-\cdots-\frac{k-d_t}{s}} 
		 + e^{-\frac{d_m-2-d_t}{s}-\cdots-\frac{k-d_t}{s}} + \cdots + e^{-\frac{k-d_t}{s}} \right) \right) \\
		&\phantom{=} - \sum_{k = d_o + 1}^{d_m} \left( 1 + 2\left( e^{-\frac{d_m-1-d_t}{s}-\frac{d_m-2-d_t}{s}-\cdots-\frac{k-d_t}{s}} 
		 + e^{-\frac{d_m-2-d_t}{s}-\cdots-\frac{k-d_t}{s}} + \cdots + e^{-\frac{k-d_t}{s}} \right) \right) \\
		&= \sum_{k = d_t + 1}^{d_o} \left( 1 + 2\left( e^{-\frac{d_m-1-d_t}{s}-\cdots-\frac{k-d_t}{s}} 
		+ e^{-\frac{d_m-2-d_t}{s}-\cdots-\frac{k-d_t}{s}} + \cdots + e^{-\frac{k-d_t}{s}} \right) \right) \\
		&= (d_o - d_t) + 2 \sum_{k = 1}^{d_o - d_t} \left( e^{-\frac{k}{s}} + e^{-\frac{k}{s}-\frac{k+1}{s}} + \cdots 
		  + e^{-\frac{k}{s} -\frac{k+1}{s} \dots - \frac{d_m-1-d_t}{s}} \right) \\ 
		&= (d_o-d_t) + 2 \left( e^{-\frac{1}{s}} \frac{1 - e^{-\frac{d_o-d_t}{s}}}{1 - e^{-\frac{1}{s}}} + e^{-\frac{3}{s}} \frac{1 - e^{-\frac{2(d_o-d_t)}{s}}}{1 - e^{-\frac{2}{s}}} 
		 + e^{-\frac{6}{s}} \frac{1 - e^{-\frac{3(d_o-d_t)}{s}}}{1 - e^{-\frac{3}{s}}} + \cdots \right).
	\end{align*}
\end{proof}

\Cref{fig:hittingtimes} shows that the analytic expression for the hitting time
is very close to the empirical average hitting time if enough terms are included
in the sum in \cref{eq:sum}. This again shows that the Markov model
approximates the actual dynamics of edit distances. For the analytic hitting
times, we kept terms in the sum larger than machine precision;
this gave 9 terms for rate $s=1$ and 27 terms for rate $s=10$. \rev{\Cref{fig:hittingtimes} also shows how differently our model behaves for different $s$ in terms of the spread of actual hitting times. \Cref{fig:hittingtimesformula} quantitatively shows how the expected hitting time grows as a function of $s$ and how the number of terms needed in \cref{eq:sum} to reach machine precision grows as a function of $s$. If $s$ is relatively small, the terms in \cref{eq:sum} decay rapidly and the hitting time can be accurately approximated with a summation of a few terms, regardless of the target or initial edit distance.}


\begin{figure}
	\centering
	\includegraphics[width=0.495\columnwidth]{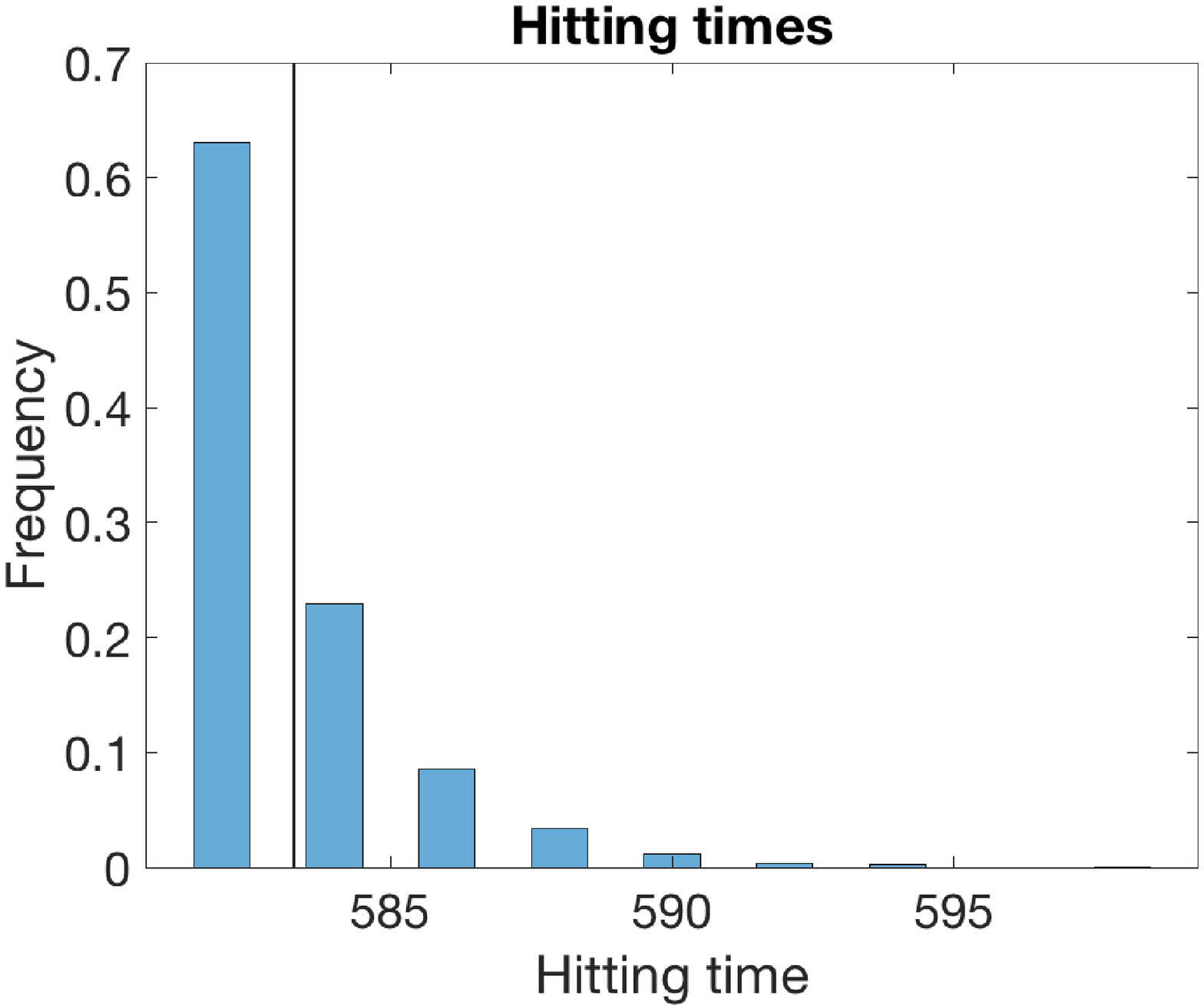} 
	\includegraphics[width=0.495\columnwidth]{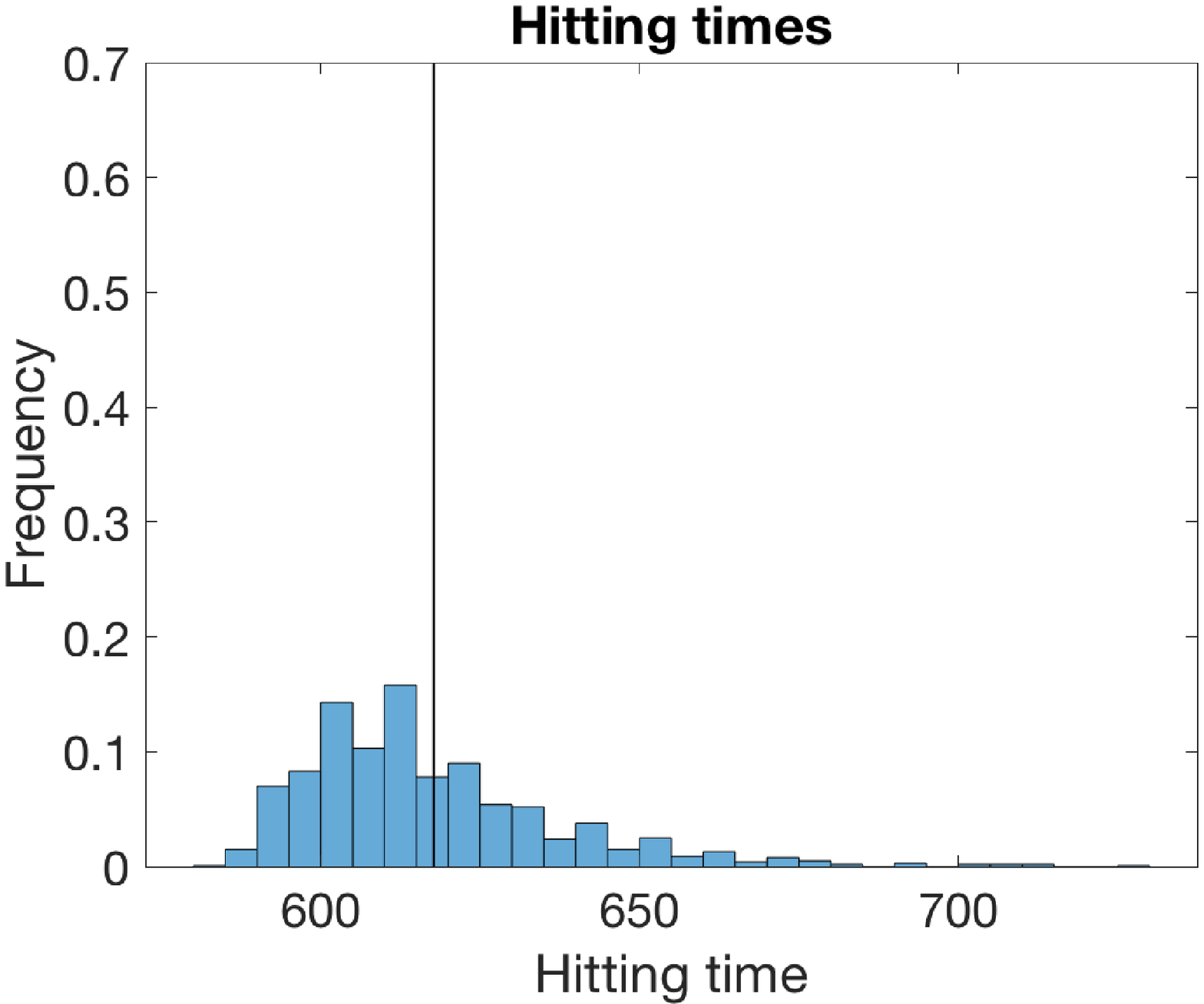}
	\caption{Histograms of hitting times of an evolution of a random graph to a 2-community target graph. Each histogram corresponds to a rate parameter ($s=1$, left; and $s=10$, right) and is based on 1000 trials. 
	The corresponding analytic hitting times from \cref{thm:hittingtimes} 
	are shown in black lines; the empirical mean hitting times differ from them by less than 0.1\%. The target distance $d_t$ is set to 10.}
	\label{fig:hittingtimes}
\end{figure}

\begin{figure}
	\centering
	\includegraphics[width=0.495\columnwidth]{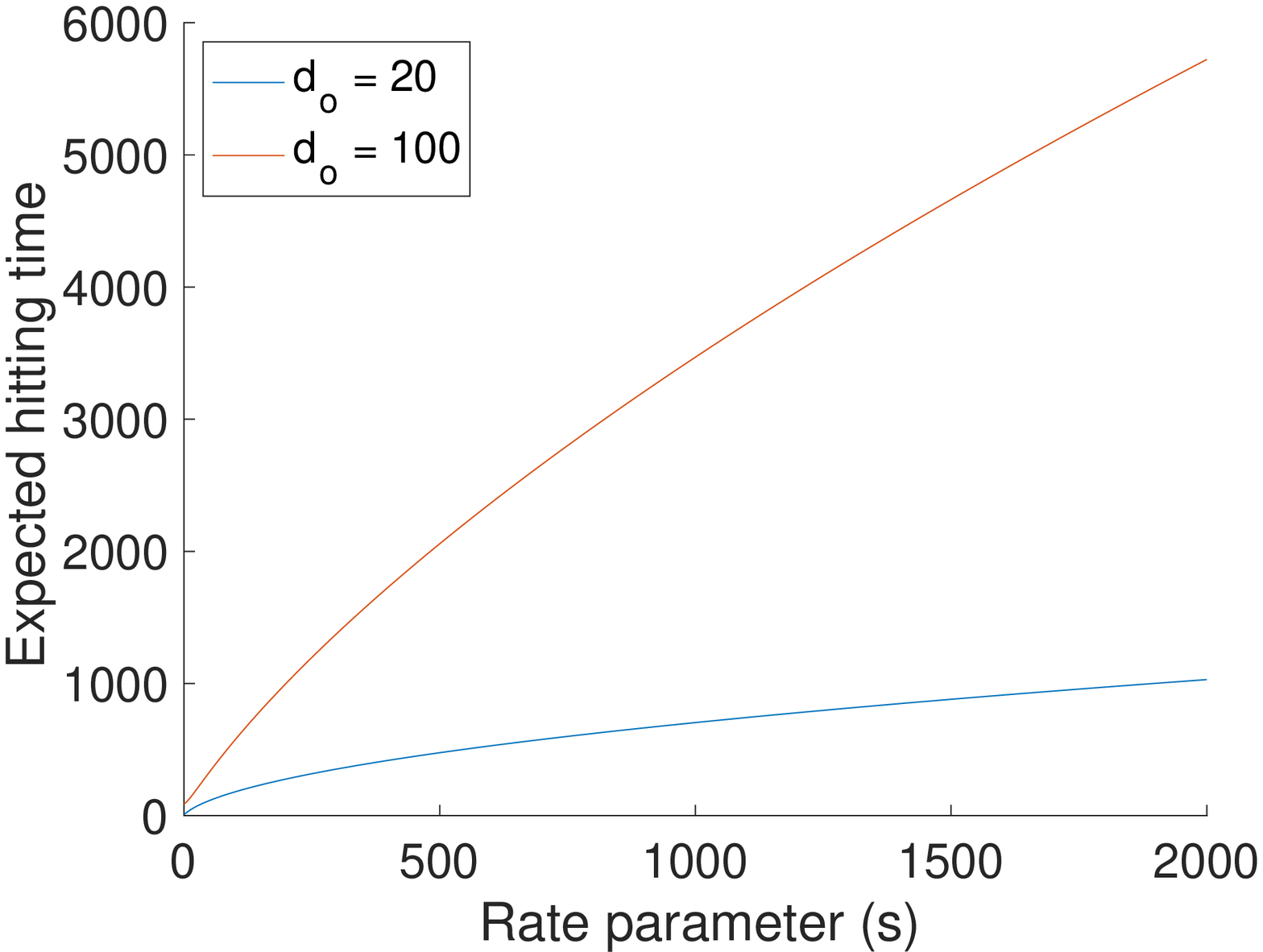} 
	\includegraphics[width=0.495\columnwidth]{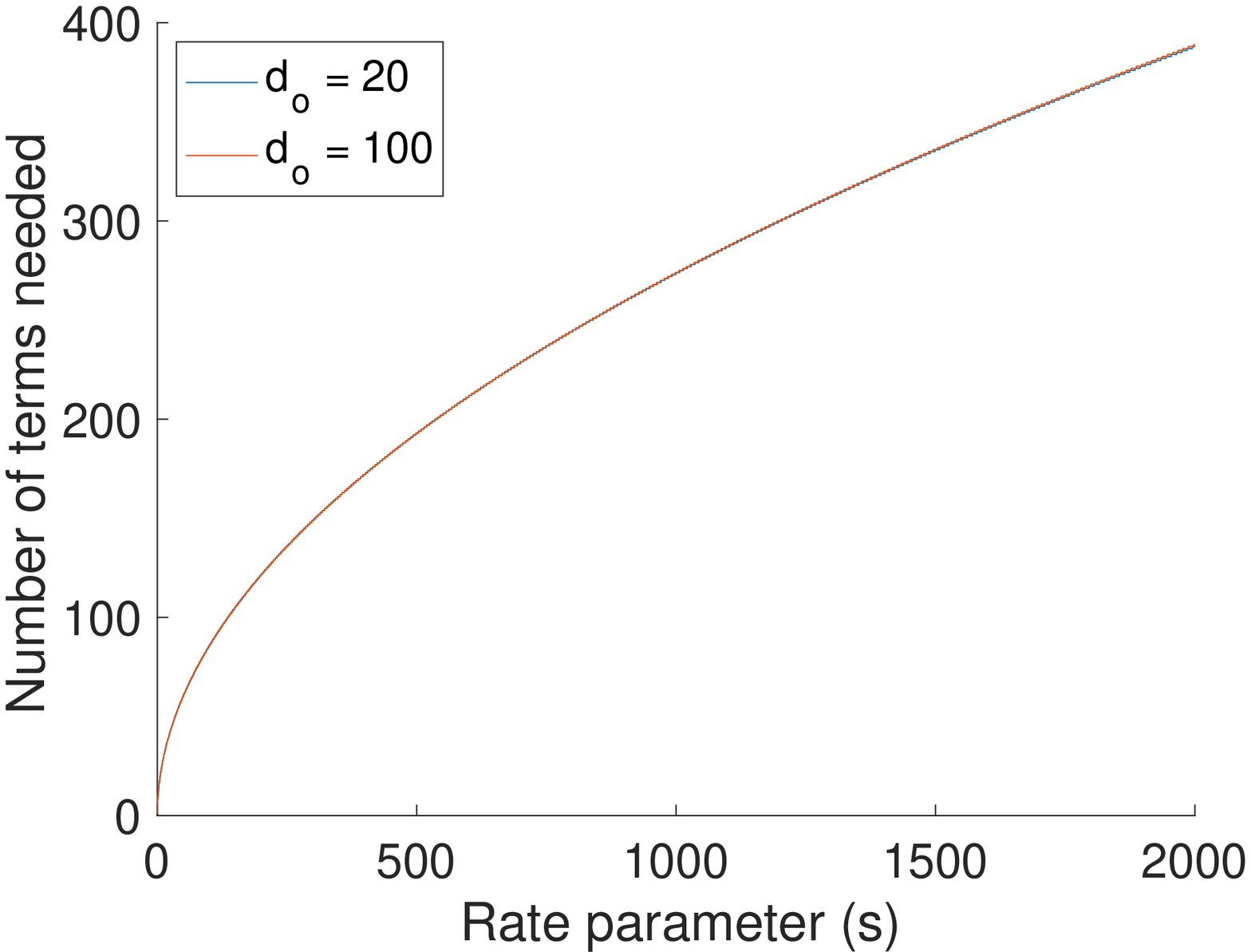}
	\caption{\rev{Expected hitting time (left) and number of terms needed in the sum to reach machine precision (right) according to the analytic expression in \cref{thm:hittingtimes}, plotted as a function of the rate parameter $s$. The target edit distance $d_t$ is set to 10, and two different initial edit distances $d_o$ to the target graph are considered ($d_o = 20$ and $d_o=100$).}}
	\label{fig:hittingtimesformula}
\end{figure}

\subsubsection{The Markov chain on the space of graphs}
The previous two results concerned a Markov chain on the space of edit
distances, but our model is actually a Markov chain on the space of all possible
graphs with a fixed number of vertices. Here we analyze the
set of graphs that the model traverses.
The exact transition probabilities between the graphs depends on how the
advancing and regressing moves are chosen at each time step, but we can
make some general statements. An $n$-vertex graph $G$ can be viewed as a
$\binom{n}{2}$-dimensional boolean vector indexed by pairs of nodes,
where each entry of the vector indicates whether $G$ contains the
corresponding edge of the complete graph. Thus the set of graphs can be viewed
as the vertices of the $\binom{n}{2}$-dimensional hypercube, and the Markov
chain on graphs can be viewed as a (biased) random walk on this hypercube. The
edit distance between two graphs is the Hamming distance between their boolean
vector representations, and we say that two graphs are adjacent if the edit distance between them is 1. Let the $i$th \emph{layer} of the hypercube be the set
of all graphs that are edit distance $i$ from the target graph: the target graph
forms layer 0, the graphs formed by deleting an edge or adding an edge to the
target graph form layer 1, and so on. The farthest layer from the target graph is layer $\binom{n}{2}$, which consists of a single graph whose edges are precisely the ones that do not exist in the target graph.

\begin{proposition}
	\label{thm:markovongraphs}
	Suppose that each advancing (resp.\ regressing) move is chosen uniformly at
	random from the set of all possible advancing (resp.\ regressing) moves. Then the
	time-averaged limiting distribution is uniform on all graphs in the same layer.
	
	Suppose that each advancing (resp.\ regressing) move is chosen uniformly at random from the set of all possible advancing (resp.\ regressing) moves, except that no false edges are allowed. Then the time-averaged limiting distribution is supported on the subgraphs of the target graph and is uniform on all such graphs that are in the same layer.
\end{proposition}
\begin{proof}
	Note that a graph in layer $i$ is adjacent to $i$ graphs in layer $i - 1$ and $\binom{n}{2}-i$ graphs layer $i + 1$. Thus, if $p_i$ is the probability that a graph in layer $i$ advances to a graph in the next lowest layer, then the Markov chain described in the first part of the statement can be described
	as follows: for $0 < i < \binom{n}{2}$, each graph in layer $i$ has a
	probability $p_i / i$ of transitioning to an adjacent graph in layer $i - 1$
	and probability $(1-p_i) / \left( \binom{n}{2}-i \right)$ of transitioning to an adjacent graph in layer $i + 1$.
	(The graph in layer $0$ has probability $1 / \binom{n}{2}$ of transitioning to each
	graph in layer $1$, and the graph in layer $\binom{n}{2}$ has a probability
	$1 / \binom{n}{2}$ of transitioning to each graph in layer $\binom{n}{2}-1$.)  
	The conclusion of the first part of the proposition holds
	because there is an automorphism of the Markov chain that sends a given vertex
	in a layer to any vertex in the same layer. (Any permutation of the vertices in
	layer 1 defines a permutation of the indices of each boolean vector and thus
	defines an automorphism of the hypercube graph.)
	For the second part, the graphs that contain edges not in the target
	graph are transient states in the Markov chain, and so in the long term the
	Markov chain looks like the chain described in the first part but supported on
	the subgraphs of the target graph.
\end{proof}

\section{Synthetic experiments}\label{sec:synthetic}
We now investigate examples of our model where the target graph is
structurally different from the starting graph. This experiment models forensic
graph analysis\textemdash given structurally distinct starting and target
graphs, can we build a sequence of graphs between them that captures the change
and subsequently analyze the behavior of that sequence to detect the structural
change? We use synthetic graphs in this section as they have especially clear structure, but we consider real-data graphs in \Cref{sec:experiments}.


We consider the stochastic block model (SBM) \cite{sbm}, a widely used model for planted graph clustering. In this model $n$ vertices are divided into clusters and two vertices are independently connected by an edge with probability $p$ if they are in the same cluster and smaller probability $q$ if they are in different clusters. We analyze two scenarios where a starting 2-block graph evolves into a target 3-block graph. This allows us to explore how community structure affects the graph transition. In the first scenario, one block in the starting 2-block graph splits into 2 equally sized blocks of half the size. In the second scenario, the 3 blocks in the target graph are independently chosen and have no a priori relation to the 2 blocks in the starting graph.

For the experiment, we took $n=120$ and set $p=0.9$ and $q=0.1$ for all the block structures. We ran both evolutions with $s = 1$ and target distance $0$ and stopped the
model as soon as it reached the target graph. (The target distance was $0$ and so we
stopped when the edit distance was first $0$.) We measured the recovery rate
(the fraction of nodes correctly classified) over time for the 3 clusters in the
target graph using a spectral clustering algorithm~\cite{Damle18}. We also
measured the \emph{subspace distance} between the space spanned by eigenvectors
associated with the three largest eigenvalues of the symmetrically normalized
adjacency matrix at each step and the analogous subspace of the symmetrically
normalized adjacency matrix of the final target graph. The subspace distance between two
subspaces represented by matrices $U$ and $V$ with orthonormal columns is
$\sqrt{1-\sigma_{\text{min}}(U^TV)^2}$~\cite{GVL}. This measures how
well the dominant invariant subspaces of the current and target graph align
and therefore if we expect a spectral algorithm to be able to pick up on the
structure of the target graph.

\Cref{fig:twotothreeblockrecovery} shows the results of this experiment, and
both scenarios display striking behavior. There is a point at which
the subspace distance sharply declines, causing the recovery rate to sharply
increase. The first scenario requires noticeably less time to detect the graph
transition, presumably because the community structure of the target graph is
related to the starting graph. 
Although the experiment is based on synthetic data, it suggests the capabilities of our model to transition from one structure to another, and the availability of algorithms to detect structural changes as they occur.


From the same experiment, we looked at the spectrum of the symmetric normalized adjacency matrix of the interpolating graphs in both scenarios and compared it to the spectrum of an algebraic interpolation between the normalized adjacency matrices of the starting and target graphs, where we ignore the graphs and simply interpolate linearly at equally spaced intervals between the two matrices. \Cref{fig:twotothreeblockeigenvalues1} shows the resulting spectra for the first scenario, and \Cref{fig:twotothreeblockeigenvalues2} shows the resulting spectra for the second scenario. In both figures, the spectra from the graph interpolation and linear matrix interpolation are strikingly similar, suggesting that our model can produce plausible spectral information while having the advantage of providing a sequence of graphs producing that spectral information. This is valuable since eigenvalue interpolation is in general known to be a difficult problem, especially in situations where the eigenvalues cross as in the second scenario. 

\begin{figure}
	\includegraphics[width=0.495\columnwidth]{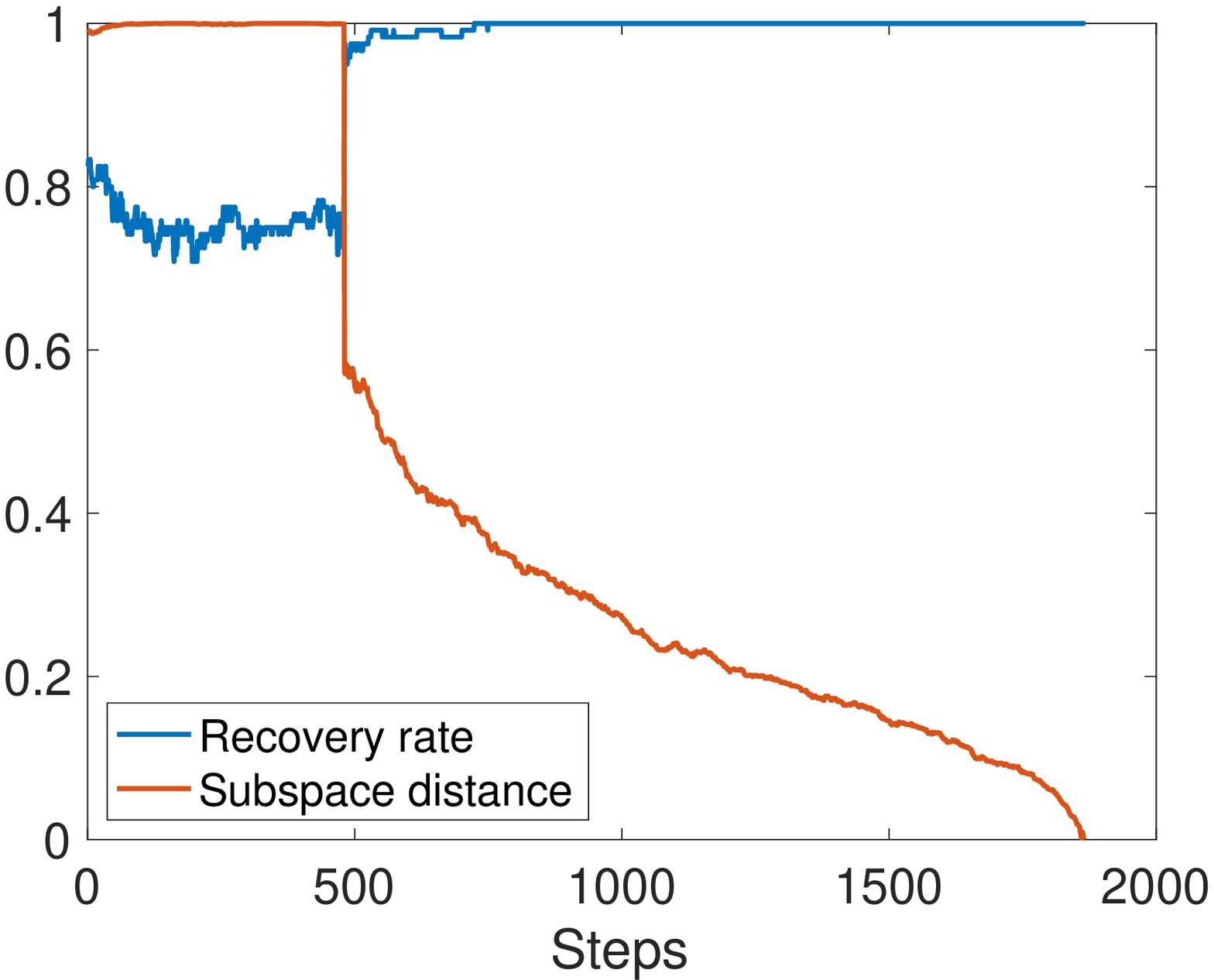}
	\includegraphics[width=0.495\columnwidth]{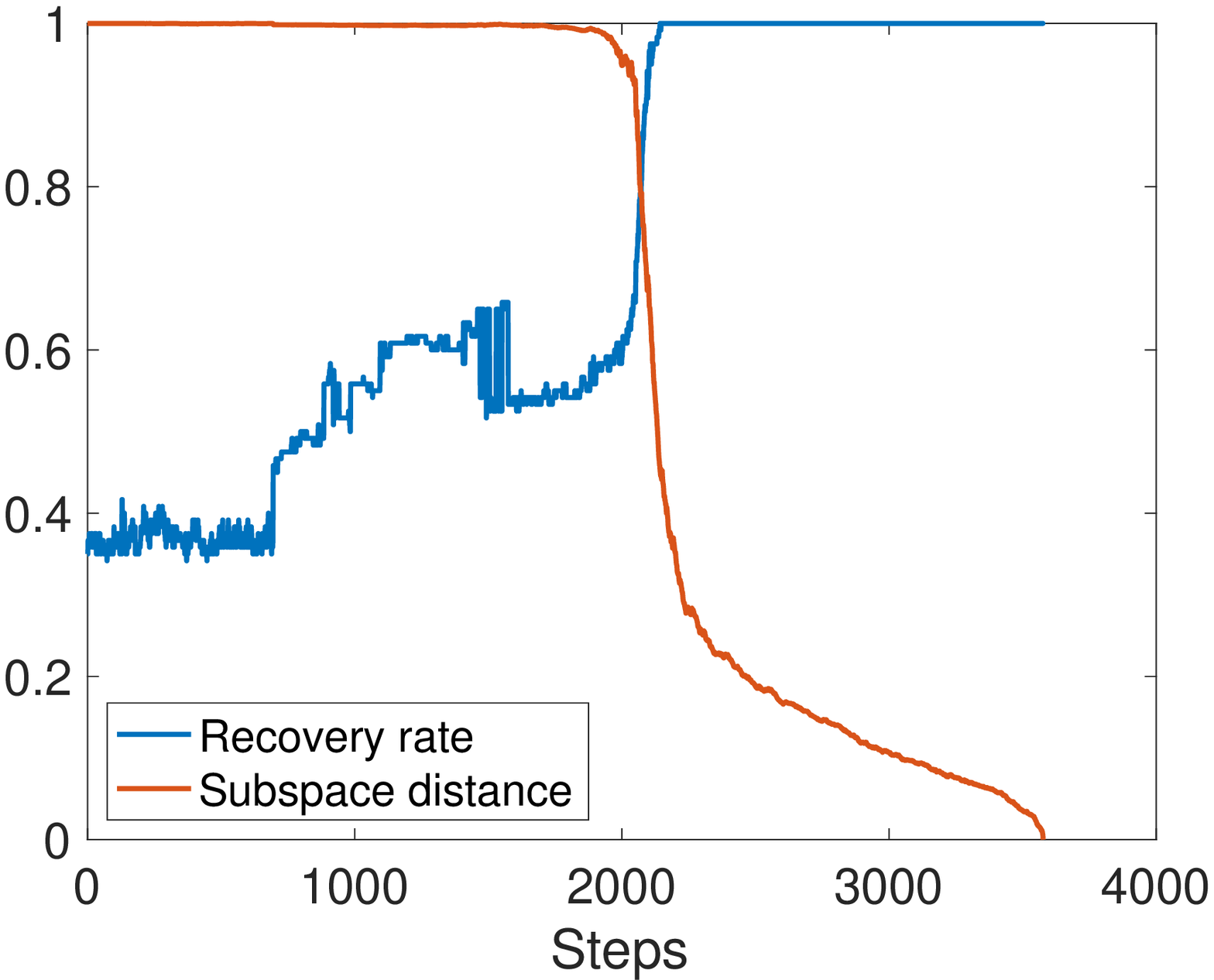}
	\caption{Cluster recovery rate of a 2-block graph evolving into a
		3-block graph in two different scenarios. The recovery rate under
		spectral clustering is shown in blue, while the subspace
		distance between the current eigenspace and the final eigenspace is
		shown in red. Left: one block splits into 2 blocks.
		Right: The 3 blocks in the target graph are independent of the 2 blocks in the starting graph.}
	\label{fig:twotothreeblockrecovery}
\end{figure}

\begin{figure}
 	\includegraphics[width=0.495\columnwidth]{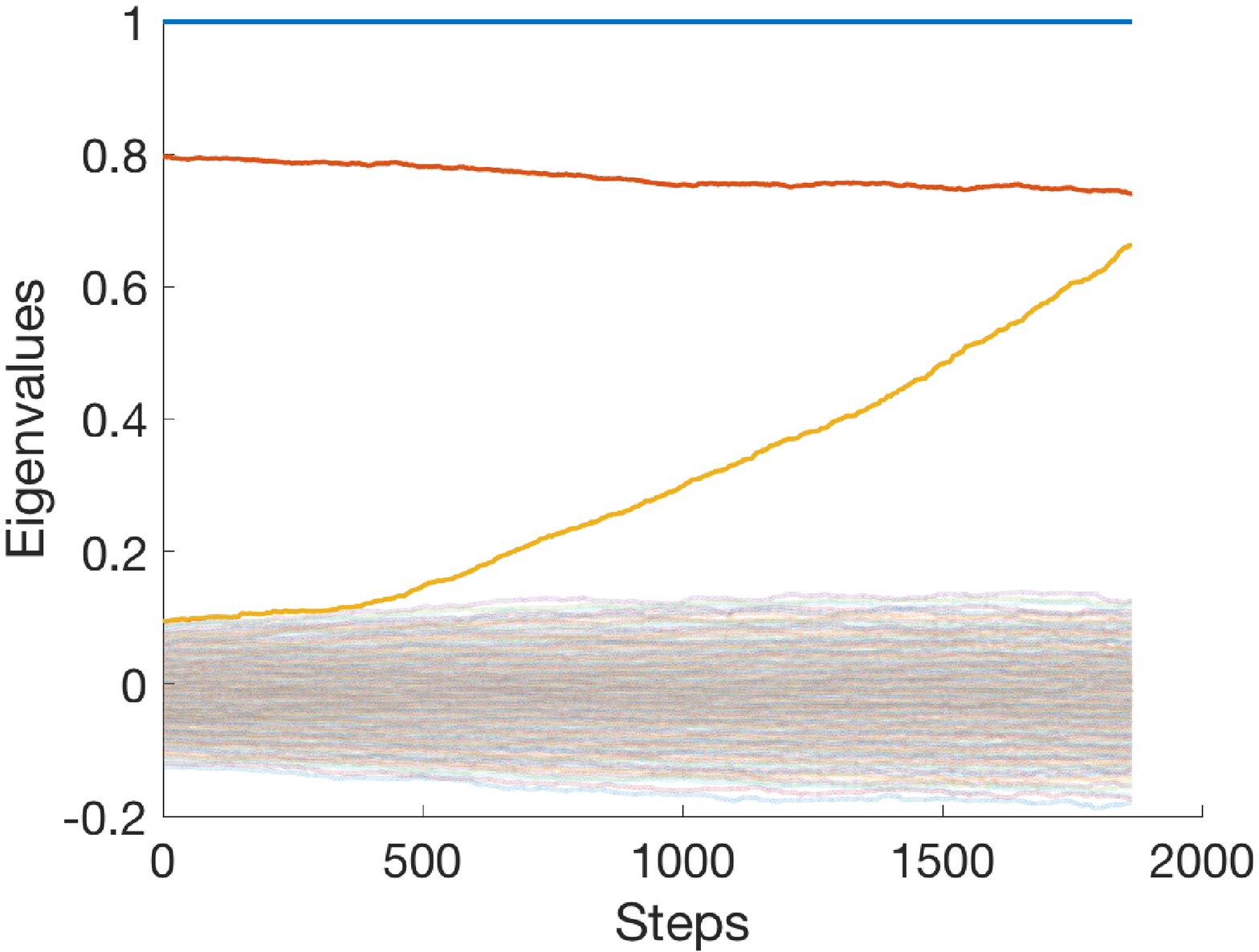}
	\includegraphics[width=0.495\columnwidth]{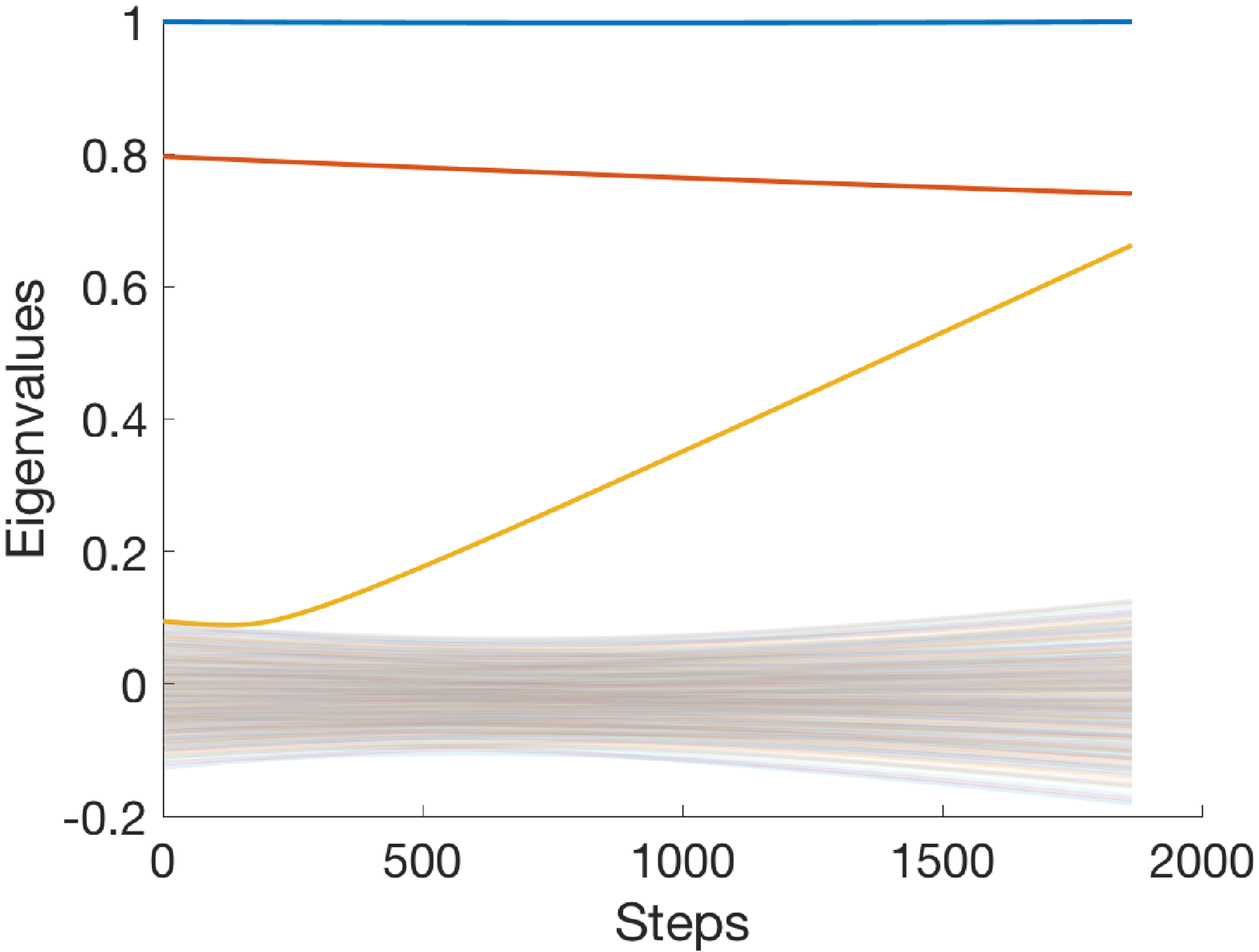}
	\caption{Spectrum of the symmetrically normalized adjacency matrix of a 2-block graph evolving into a 3-block graph, where one block splits into 2 blocks. 
	Left: Graph interpolation.
	Right: Equispaced linear interpolation between matrices. The full spectrum is computed at every step of the graph interpolation (resp.\ every increment of the equispaced linear interpolation), and the most significant eigenvalues are highlighted to distinguish them from the bulk.}
	\label{fig:twotothreeblockeigenvalues1}
\end{figure}

\begin{figure}
	\includegraphics[width=0.495\columnwidth]{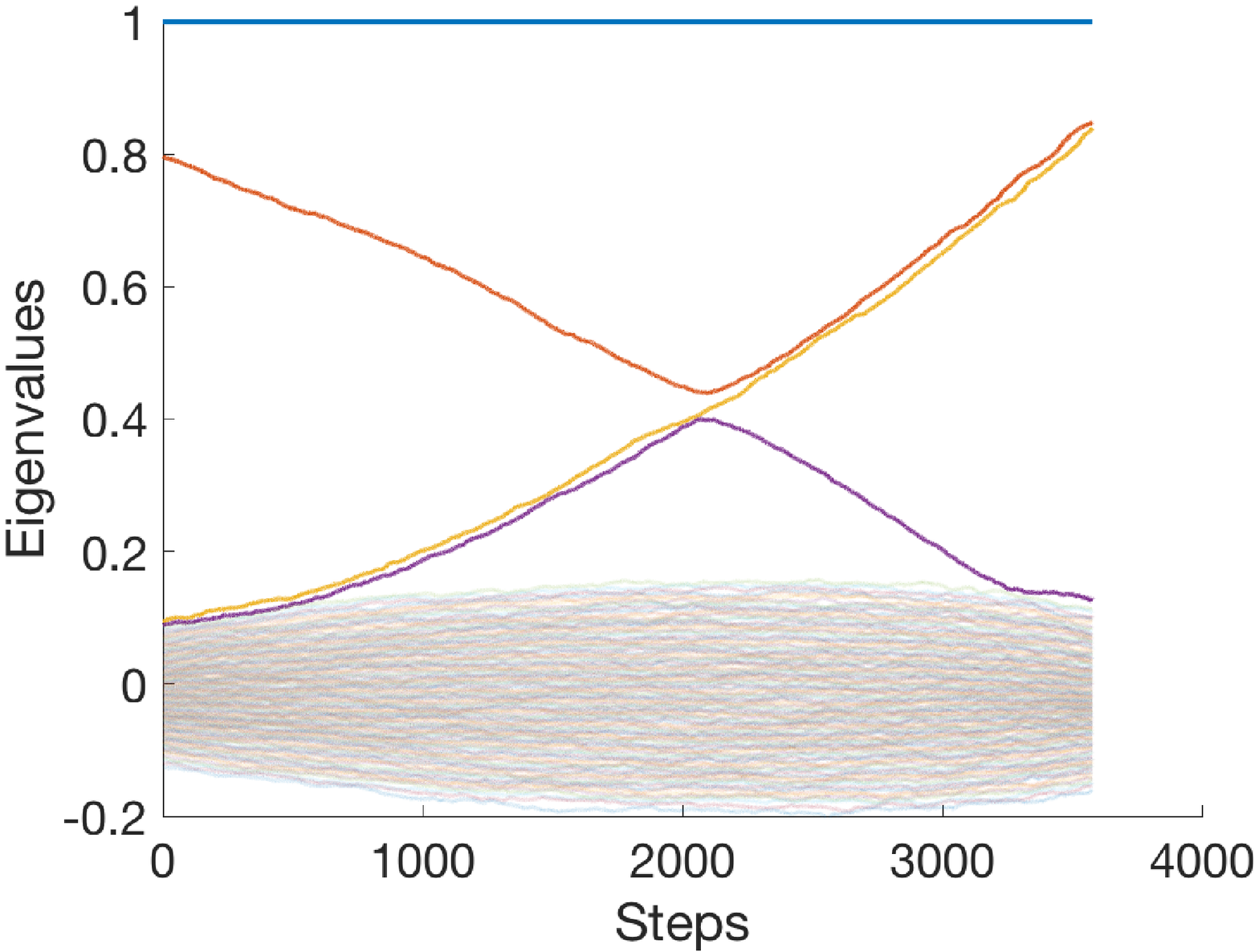}
	\includegraphics[width=0.495\columnwidth]{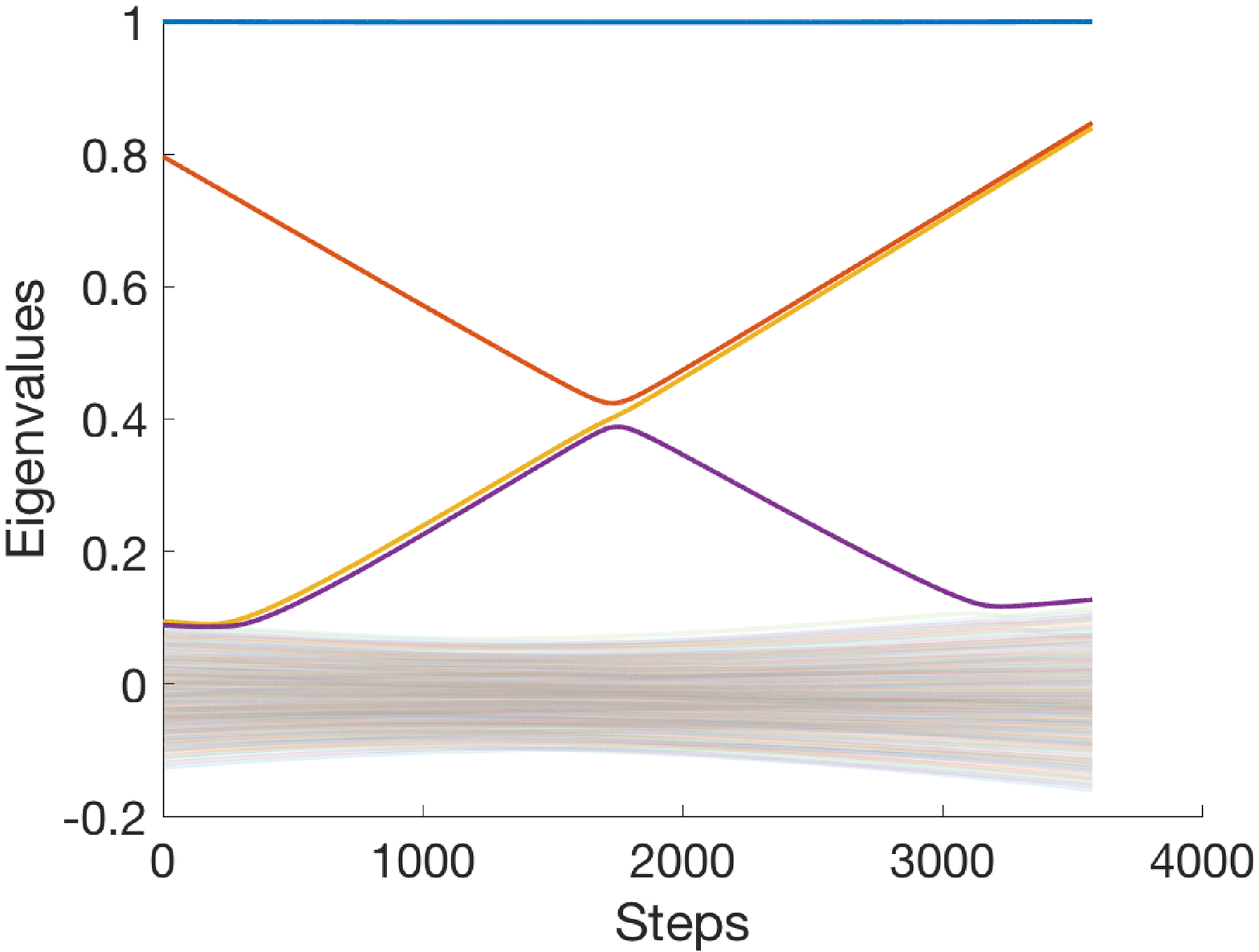}
	\caption{Spectrum of the symmetrically normalized adjacency matrix of a 2-block graph evolving into a 3-block graph, where the 3 blocks in the target graph are independent of the 2 blocks in the starting graph.
	Left: Graph interpolation.
	Right: Equispaced linear interpolation between matrices. The full spectrum is computed at every step of the graph interpolation (resp.\ every increment of the equispaced linear interpolation), and the most significant eigenvalues are highlighted to distinguish them from the bulk.}
	\label{fig:twotothreeblockeigenvalues2}
\end{figure}

\section{Real-data experiments}
\label{sec:experiments}

We now interpolate between snapshots of several real-world networks,
demonstrating that our model can provide a sensible graph
interpolation where graphs and their properties vary ``smoothly,'' even in the absence of any data between the snapshots.
To interpolate between two graphs, we set the first graph to be the starting graph
and the second graph to be the target graph. We then set the target edit
distance to be $0$ and ran our interpolation with a specified rate
parameter until the target graph is reached.
To interpolate between multiple graphs, we do this procedure for the first and
second graphs, then the second and third graphs, and so on. We then compare our
interpolation with extrapolation methods (growth models) as described below. 
\Cref{tab:datasets} summarizes the datasets we use.

\begin{table}
	\centering
	\caption{Summary statistics of undirected snapshot graph datasets.}
	\label{tab:datasets}
	\begin{tabular}{r c l l}
		\toprule
		Dataset & \# snapshots & \# nodes & \# edges (max) \\
		\midrule
		CollegeMsg \cite{Panzarasa-2009-CollegeMsg} & 2 & 1899 & 5852 \\
		\rev{email-Eu-core-temporal \cite{Paranjape-2017-motifs}} & \rev{10} & \rev{1005} & \rev{16064} \\
		van de Bunt students \cite{VanDeBunt-1999-friendship} & 7 & 32 & 59  \\
		DBLP coauthorship \cite{Ley-2009-DBLP} & 9 & 554885 & 407575 \\
		\bottomrule
	\end{tabular}
\end{table}

In our first type of example, we consider a strictly growing network and compare
our interpolation with various extrapolations available from network growth models. For
this purpose, we used the first two datasets shown in
\cref{tab:datasets}. The first one is based on private messages on a college
messaging network \cite{Panzarasa-2009-CollegeMsg}. 
The data consists of timestamped edges indicating communication between two individuals. From this, we create a growing
dynamic graph based on the accumulation of these messages over a period of time. We took two
snapshots: the first snapshot is the empty graph and the second
snapshot is the aggregated graph of all communications within the first 30 days of the
dataset.

\begin{figure}
	\begin{subfigure}{0.495\columnwidth}
		\includegraphics[width=0.95\linewidth]{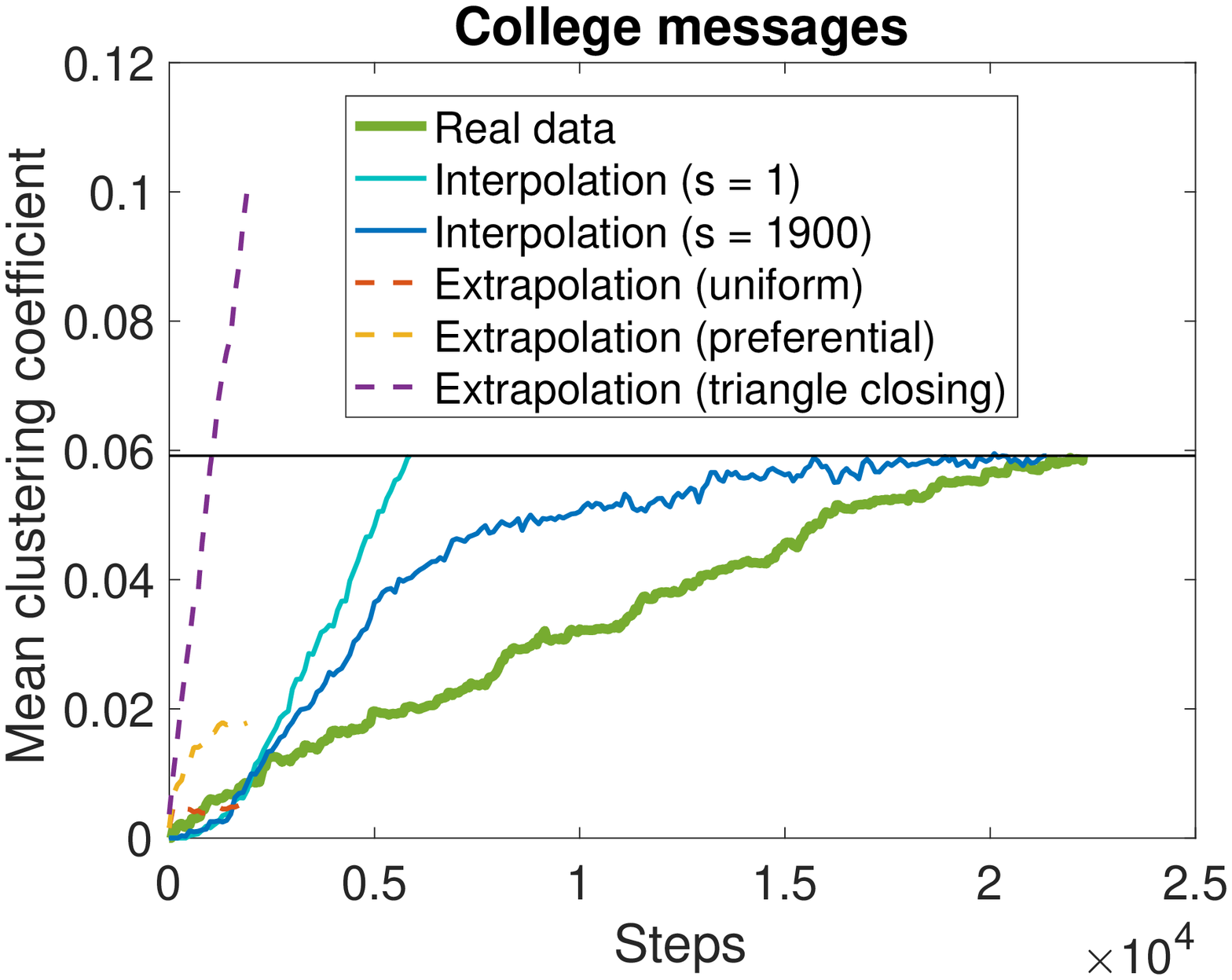}
	\end{subfigure}%
	\begin{subfigure}{0.495\columnwidth}
		\includegraphics[width=0.95\linewidth]{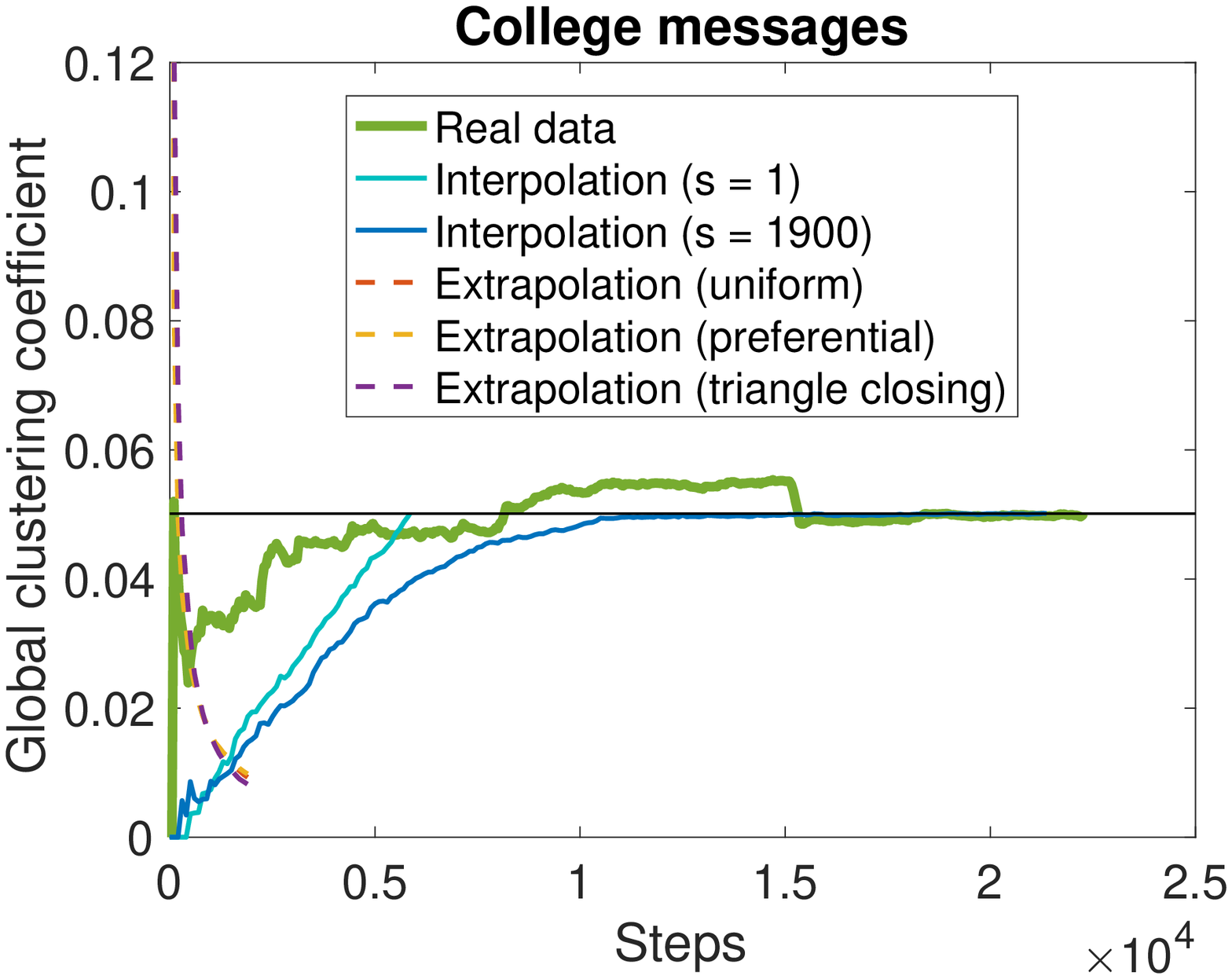}
	\end{subfigure}
	\caption{Evolution of mean and global clustering coefficient in a college message network. 
	The actual evolution is highlighted in \rev{green} along with the interpolation and three extrapolation methods. 
	The horizontal line indicates the clustering coefficient of the final snapshot. 
	The interpolation with rate parameter $s=1900$ corresponds most closely with the real data.
	}    
	\label{fig:collegemsg}
\end{figure}

\Cref{fig:collegemsg} shows the change in the mean and global clustering coefficient in
the actual data over time compared with our graph model interpolations and
three extrapolation methods. The mean clustering coefficient of a graph is the mean local clustering coefficient across all nodes,
where the local clustering coefficient of a node is the number of connected pairs of neighbors of the node divided by the total number of pairs of neighbors
(and is set to 0 if the node has degree 0 or 1). 
The global clustering coefficient of a graph is
 the number of closed length-2 paths divided by the total number of length-2 paths in the entire graph.

The interpolations in \Cref{fig:collegemsg} uses rates $s=1$ and
$s=1900$; the latter was chosen to match most closely with the real data. To elaborate further, we actually know that the real dynamic network took 21812 steps between the snapshots, so using \cref{thm:hittingtimes} we took $s$ to be the multiple of $50$ such that the expected hitting time was closest to the number of steps taken.
In general, a higher rate parameter increases the expected number of steps in the
interpolation, so this parameter can be tuned if an estimate of the number of edits actually taken from one graph to another is known.

The extrapolations are from
three network growth models: uniform attachment, Barab\'{a}si-Albert
preferential attachment \cite{Barabasi99}, and Jackson-Rogers triangle closing
\cite{Jackson07}. These growth models all start with a small initial graph. In
each step, a new node appears and connects to a predefined number of nodes (on
average) in the existing network according to a probabilistic distribution on
the nodes.
For uniform attachment, we start with a clique on some $m$ vertices, and each new
node connects to $m$ nodes chosen uniformly at random in the existing graph. For
preferential attachment, we start with a clique on some $m$ vertices, and each new
node connects to $m$ samples of random nodes in the existing graph, chosen with
probability proportional to their degree. 
For Jackson-Rogers triangle closing, we start with a clique on some $m_r+m_n+1$
vertices. At each step, we pick $m_r$ nodes in the existing graph uniformly at
random and call them ``parent nodes.''  The new node connects to each parent
node independently with probability $p_r$. We also pick $m_n$ nodes
uniformly at random from the parents' immediate neighbors and
connect to each of them independently with probability $p_n$. 

We set the extrapolation parameters to match approximately the average degree of the target graph. For uniform and preferential attachment, the parameter $m$ is the number of nodes that each new node connects to, and so we set $m$ to be the nearest integer to the average degree of the target graph. For the college message network in \Cref{fig:collegemsg}, $m=3$ as can be seen from \Cref{tab:datasets}. The triangle closing model has parameters $m_r$, $p_r$, $m_n$, and $p_n$. The average degree, in expectation, of a graph produced by this model is $m_rp_r + m_np_n$,
so we matched this with the average degree in the dataset. To set
a reasonable choice of parameters, we set $p_r = p_n = 0.5$ and experimentally
adjusted $m_n$ relative to $m_r$ until the mean clustering coefficient of the graph
produced by the model approximately matched that of the target graph. Here we set $m_r = 5$ and $m_n = 1$ (we found that $m_r = 6$ and $m_n = 0$ produced a final mean clustering coefficient that was as low as the uniform attachment model.)


\Cref{fig:collegemsg} shows that the extrapolation methods do not yield the correct final mean clustering coefficient, 
nor even the correct qualitative behavior for the global clustering coefficient,
instead decreasing rapidly from 1 (corresponding to the starting clique with isolated vertices). These existing methods are extrapolating from a starting graph, which is not appropriate for our task. 
This further motivates the use of our interpolation model.

\begin{figure}
	\begin{subfigure}{0.495\columnwidth}
		\includegraphics[width=0.95\linewidth]{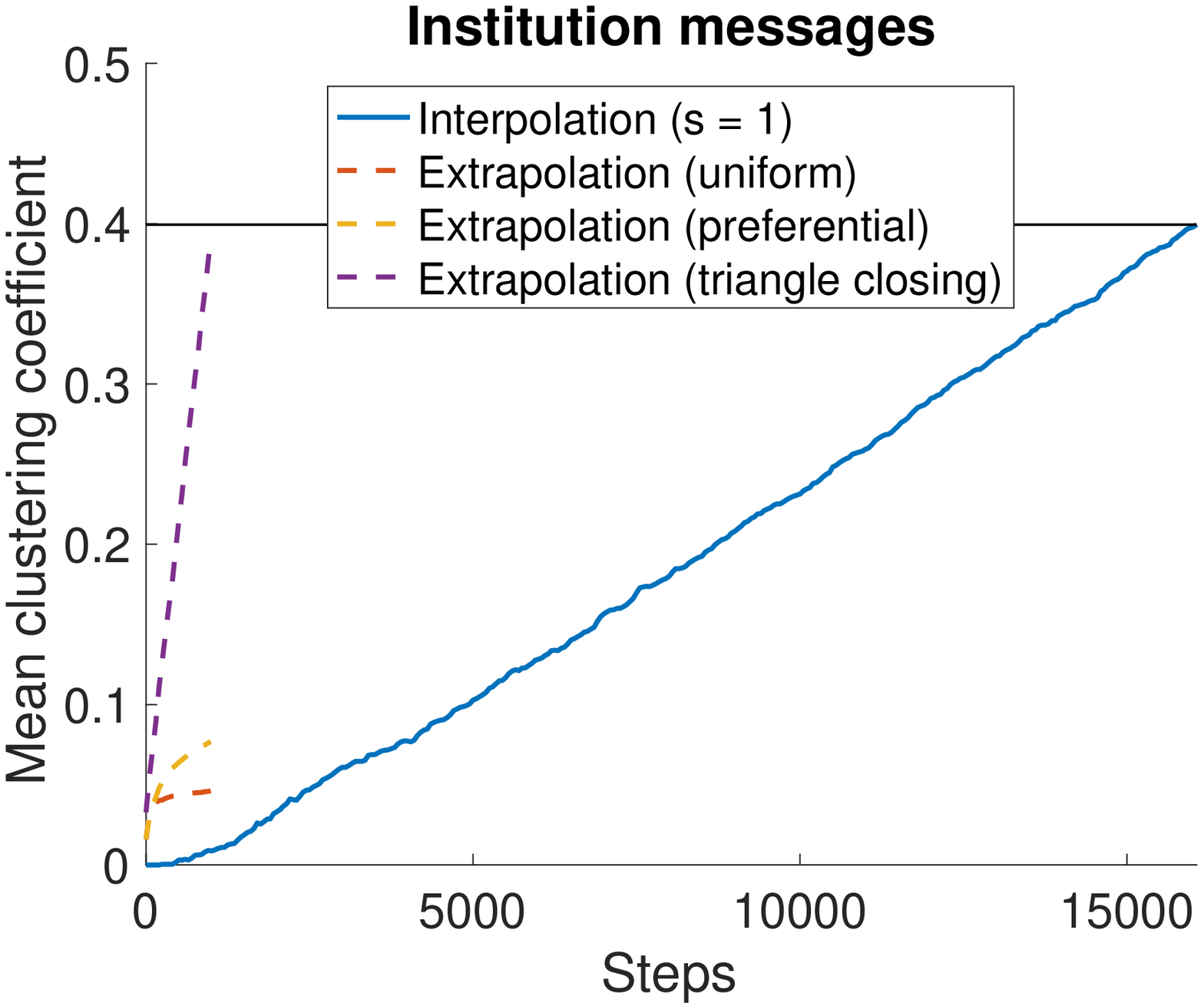}
	\end{subfigure}%
	\begin{subfigure}{0.495\columnwidth}
		\includegraphics[width=0.95\linewidth]{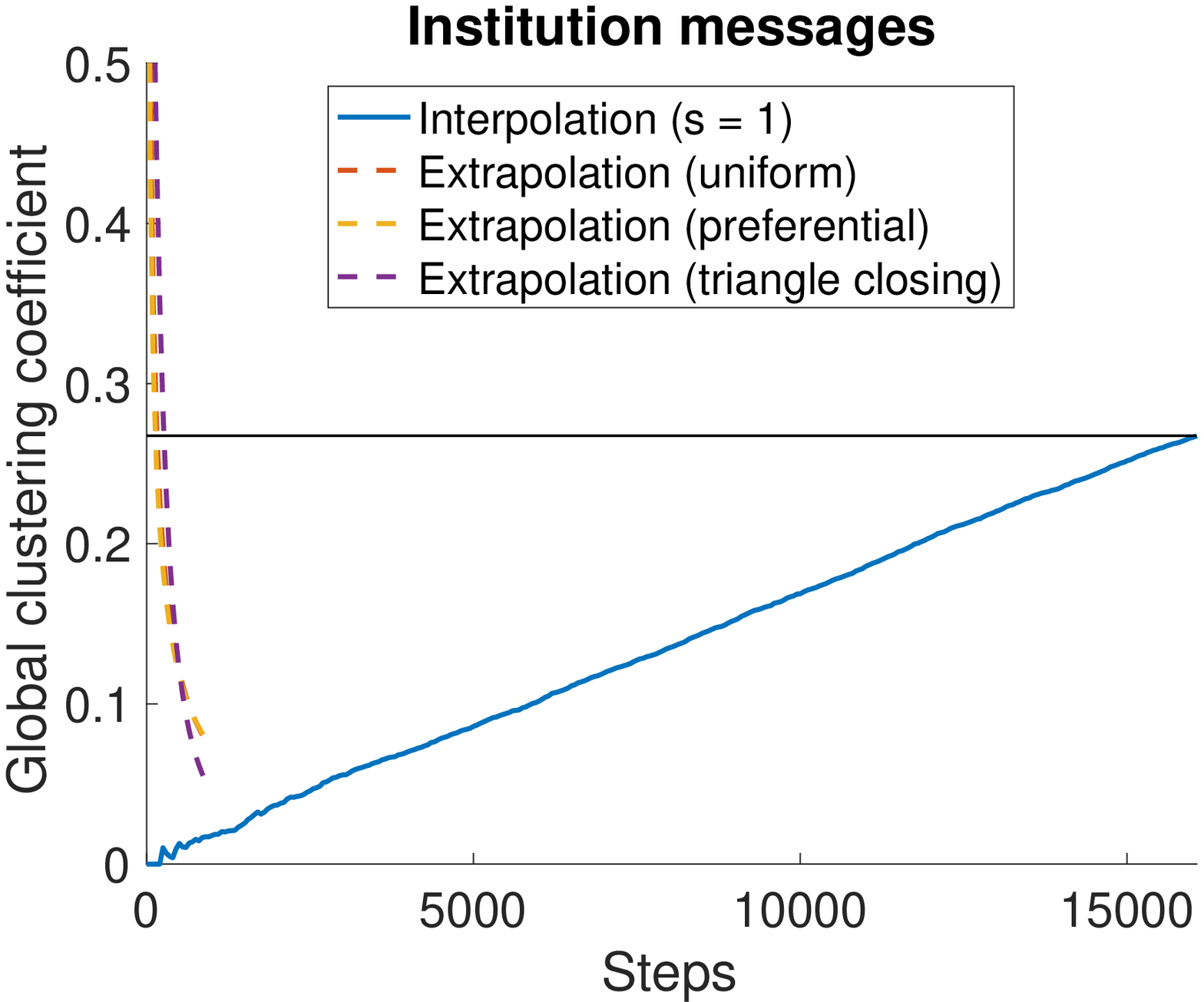}
	\end{subfigure}
	\caption{Evolution of mean and global clustering coefficient on the email-Eu-core network using interpolation and three extrapolation methods. The horizontal line is the clustering coefficient of the final snapshot. Our interpolation better and more naturally fits the data.
	}
	\label{fig:emaileucore}
\end{figure}

\begin{figure}
	\begin{subfigure}{0.495\columnwidth}
		\includegraphics[width=0.95\linewidth]{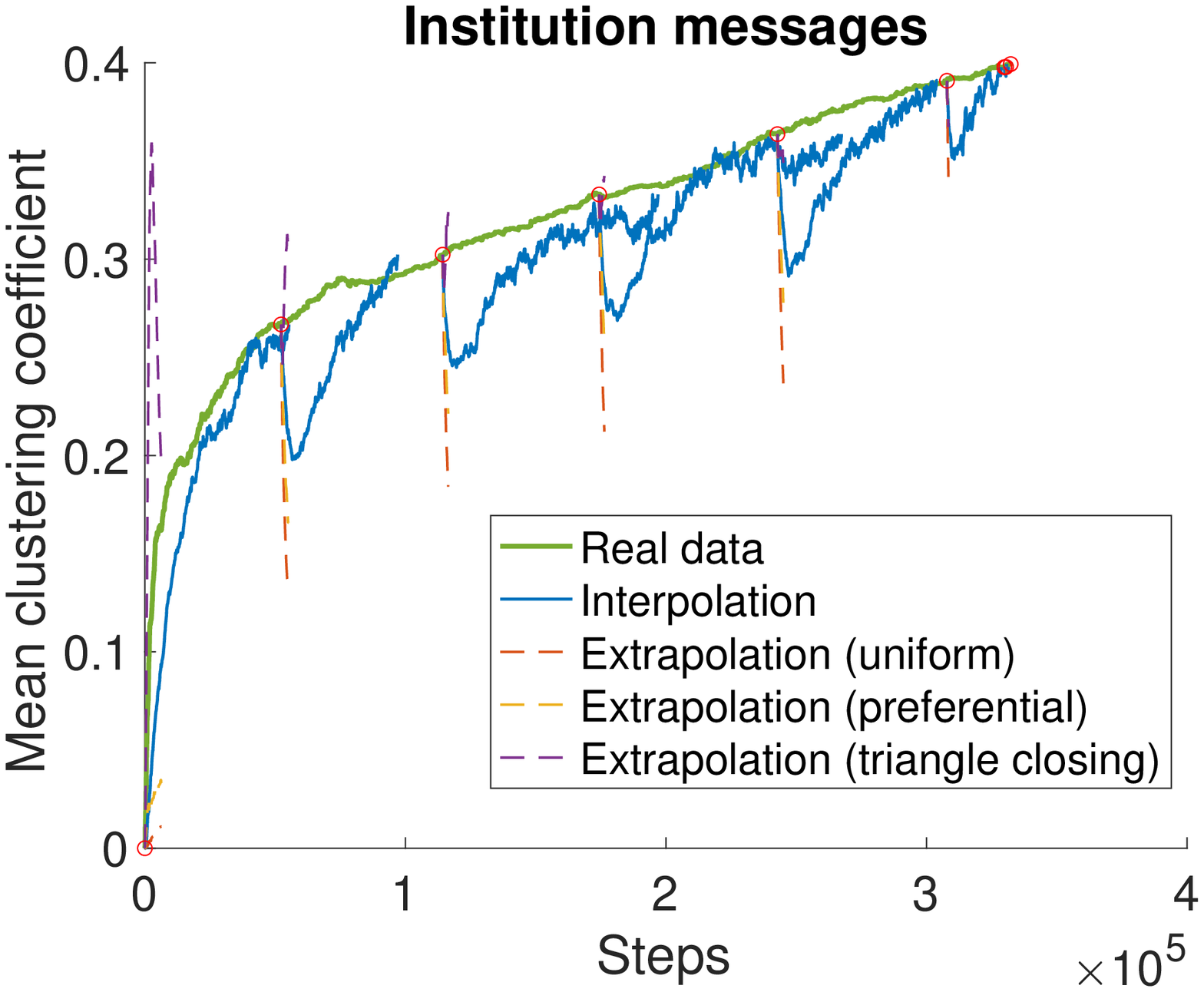}
	\end{subfigure}%
	\begin{subfigure}{0.495\columnwidth}
		\includegraphics[width=0.95\linewidth]{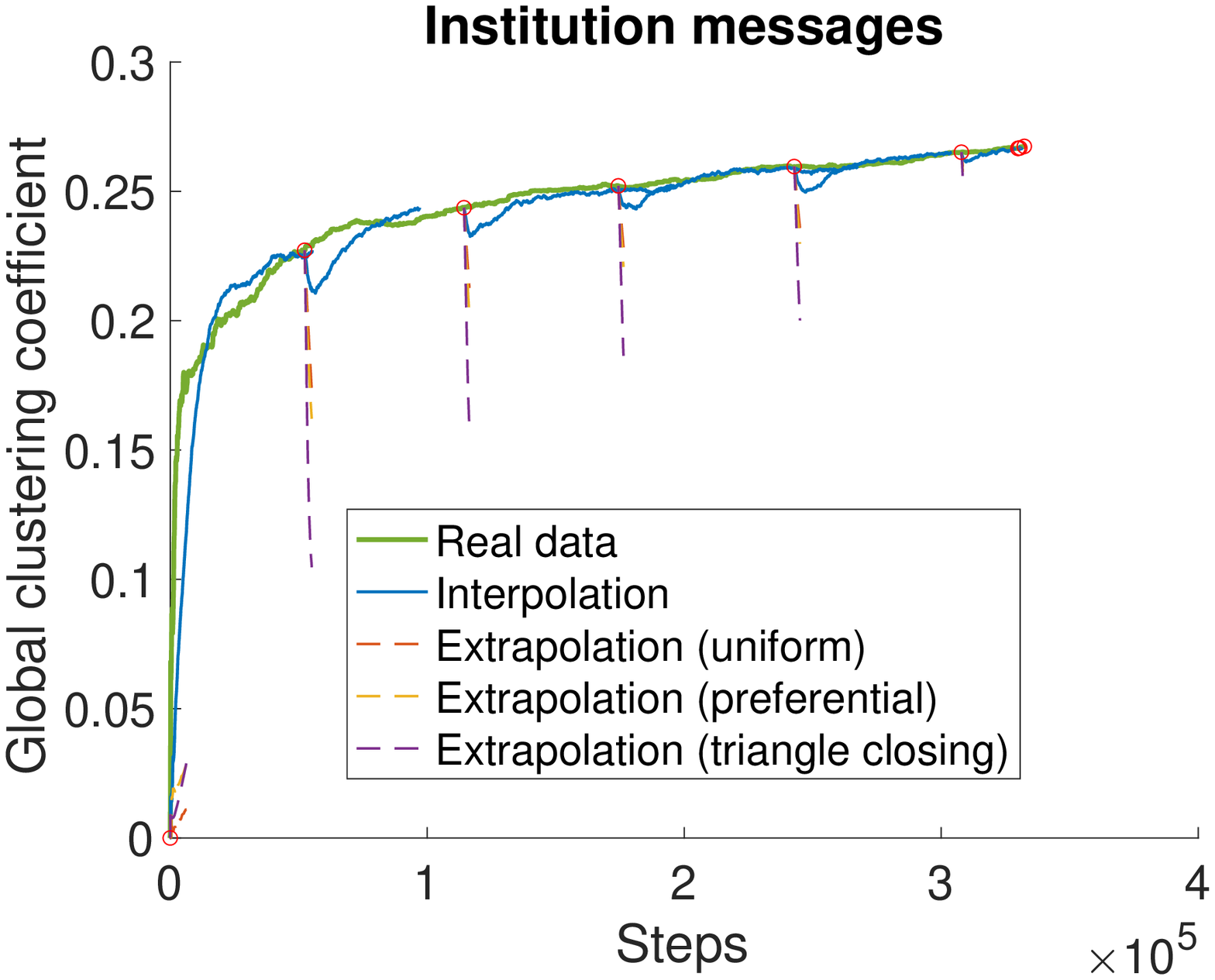}
	\end{subfigure}
	\caption{\rev{Evolution of mean and global clustering coefficient on the email-Eu-core-temporal network. The actual evolution is highlighted in green along with the interpolation and three extrapolation methods. Red circles denote clustering coefficients of the snapshots themselves. The horizontal axis is according to the steps in the real data; this explains the overlapping lines in the interpolation since the number of steps in the interpolation does not align perfectly with the real data.}
	}
	\label{fig:emaileucoretemporal}
\end{figure}

\rev{We repeated our experiments using a dataset of emails at a European research institution~\cite{Paranjape-2017-motifs} with timestamped edges indicating emails exchanged between individuals at the institution. As before, we created a growing dynamic graph based on the accumulation of the email exchanges. We did two experiments: one using just 2 snapshots (the empty graph and the final graph accumulating all the email exchanges) and one using 10 snapshots (details below).}
For the first experiment, \Cref{fig:emaileucore} shows the change in mean and global clustering coefficient for our graph model interpolation with rate $s = 1$ and for the three extrapolation methods just described. For uniform and preferential attachment, we set $m=16$, and for triangle closing, we set $p_r = p_n = 0.5$, $m_r = 6$, and $m_n = 26$. Here we were better able to fit the triangle closing model to match the final mean clustering coefficient of the target. 
However, for the global clustering coefficient, all of the extrapolation methods miss the mark.

\rev{\Cref{fig:emaileucoretemporal} shows our experiments on the same dataset, but where we took a snapshot of the growing graph roughly every 100 days to get 10 snapshots in total. (Because the interval between timestamps is far from uniform, a number of snapshots are clustered near the end of the dataset when the timestamps are ordered chronologically.) For each interpolation, we used \cref{thm:hittingtimes} to choose the best rate parameter up to a multiple of 50 since we know the edit distance between snapshots and the number of steps taken between them. For the extrapolations, we used the same models as before but added edges one at a time. As shown in \Cref{fig:emaileucoretemporal}, only our interpolation produces a plausible result with roughly the correct overall number of steps, and it also best captures the statistical trends of the data over time. One apparent inaccuracy is the dipping behavior near the beginning of each interpolation. This may be due to the assumption of uniformly random edge edits, which degrade the clustering of the graph as edges that get added near the beginning of the interpolation tend not to have any significant connection to the existing edges. This behavior is magnified in a later experiment 
(see \Cref{fig:coauth}).}

\begin{figure}
	\begin{subfigure}{0.495\columnwidth}
		\includegraphics[width=0.95\linewidth]{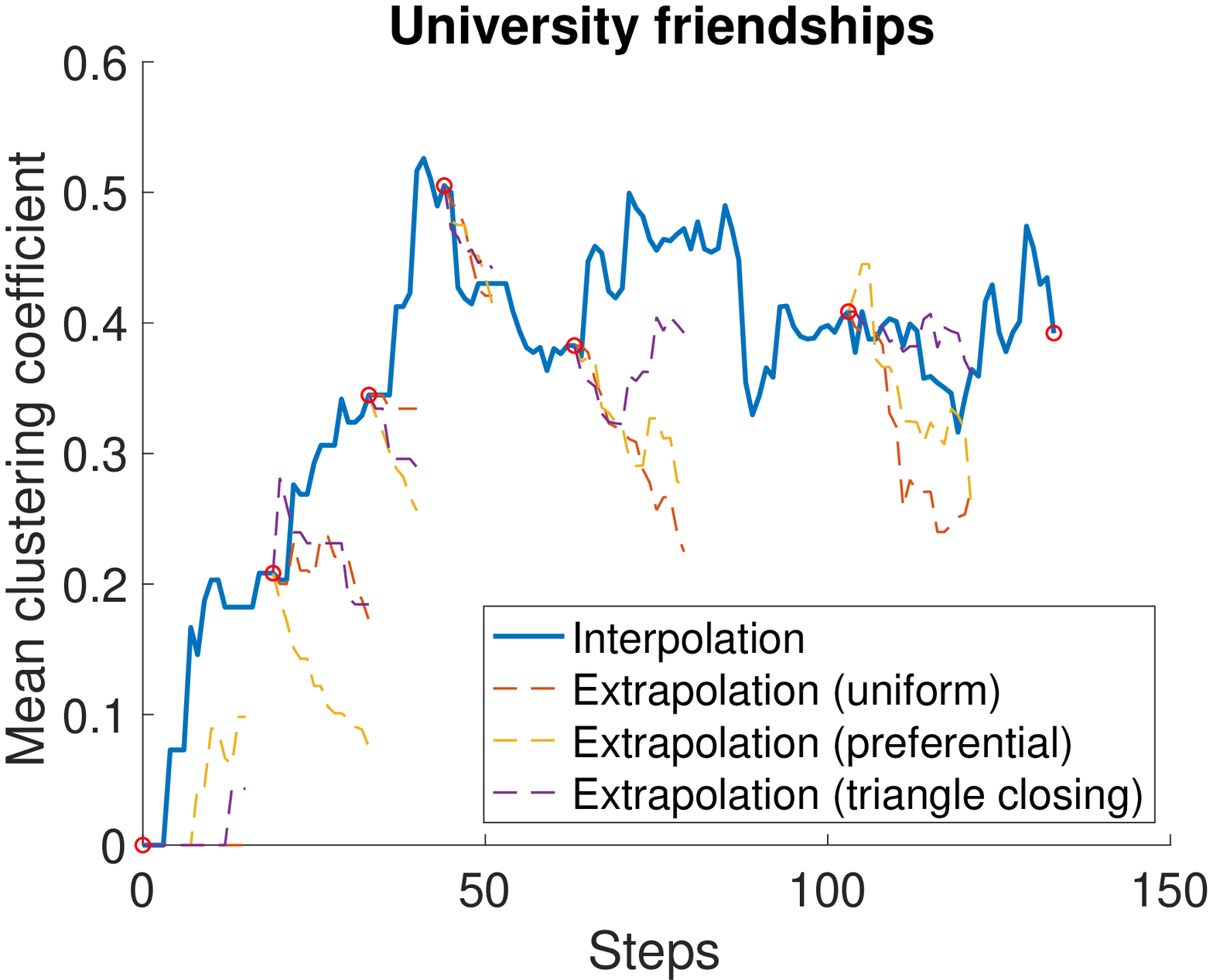}
	\end{subfigure}%
	\begin{subfigure}{0.495\columnwidth}
		\includegraphics[width=0.95\linewidth]{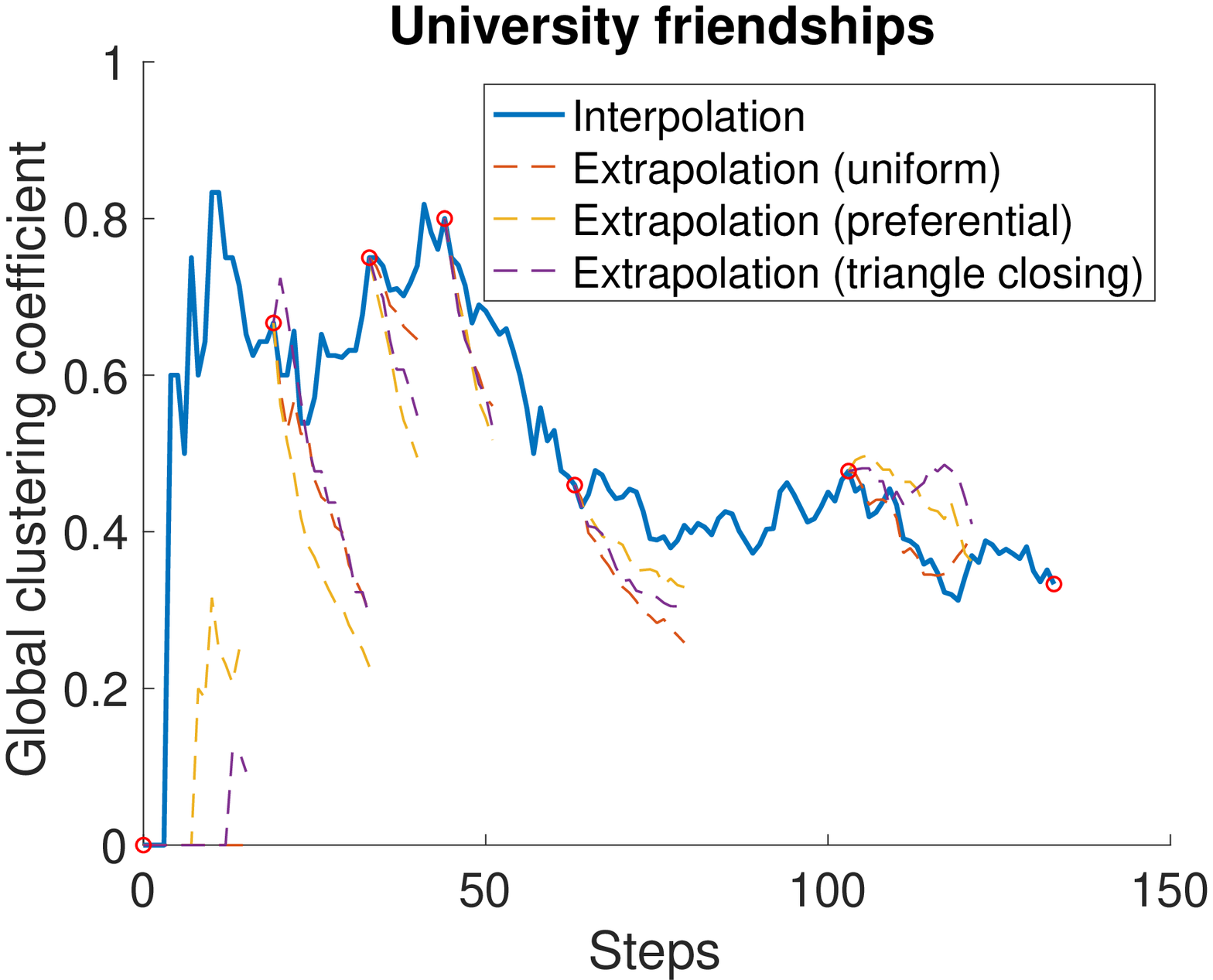}
	\end{subfigure}
	\caption{Evolution of mean and global clustering coefficient in interpolation and three extrapolations between 7 snapshots of university freshmen friendships. The rate parameter $s = 1$ for the interpolation. Red circles denote clustering coefficients of the snapshots themselves.}
	\label{fig:vdbunt}
\end{figure}

Our second type of real-world data example further shows the potential of our
model to produce a feasible interpolation between a series of snapshots. The
datasets we used here are the last two in \Cref{tab:datasets}. The first
dataset consists of 7 snapshots of self-reported friendships between 32
university freshmen in a survey-based experiment \cite{VanDeBunt-1999-friendship}. 
\Cref{fig:vdbunt} displays the
mean and global clustering coefficient over time for our interpolation with rate 1 and compares
it with the same three extrapolation models. 
Although it is less natural to use extrapolation for non-growing
networks, we included a ``decaying'' component as follows. If the number
of edges of the graph increased between snapshots, then we proceed as before with
the usual extrapolation one edge at a time until the number of edges in the extrapolation equals the number of edges in the target snapshot.
If, on the other hand, the number of edges of the graph decreased between snapshots, then in each step of the extrapolation we select a ``new node'' uniformly at random. From this ``new node,'' we choose a neighbor to sever connections with. For uniform
attachment, we delete a uniformly random neighbor; for preferential
attachment, we delete a random neighbor inversely weighted by degree;
and for triangle closing, we delete a random neighbor inversely
weighted by the number of triangles in which the neighbor participates. 

\begin{figure}
	\begin{subfigure}{0.495\columnwidth}
		\includegraphics[width=0.95\linewidth]{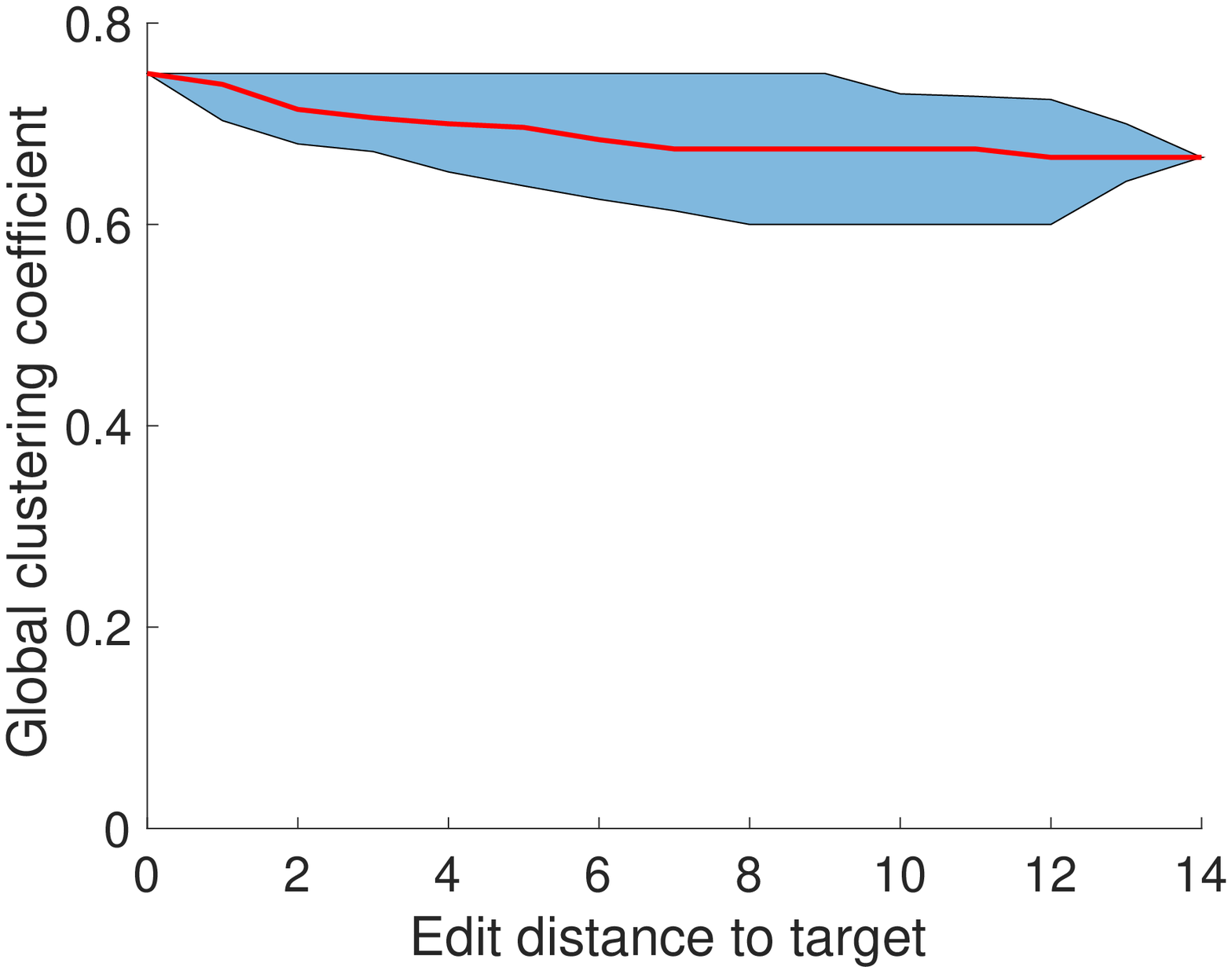}
		\caption{Snapshot 2 to 3.}
	\end{subfigure}%
	\begin{subfigure}{0.495\columnwidth}
		\includegraphics[width=0.95\linewidth]{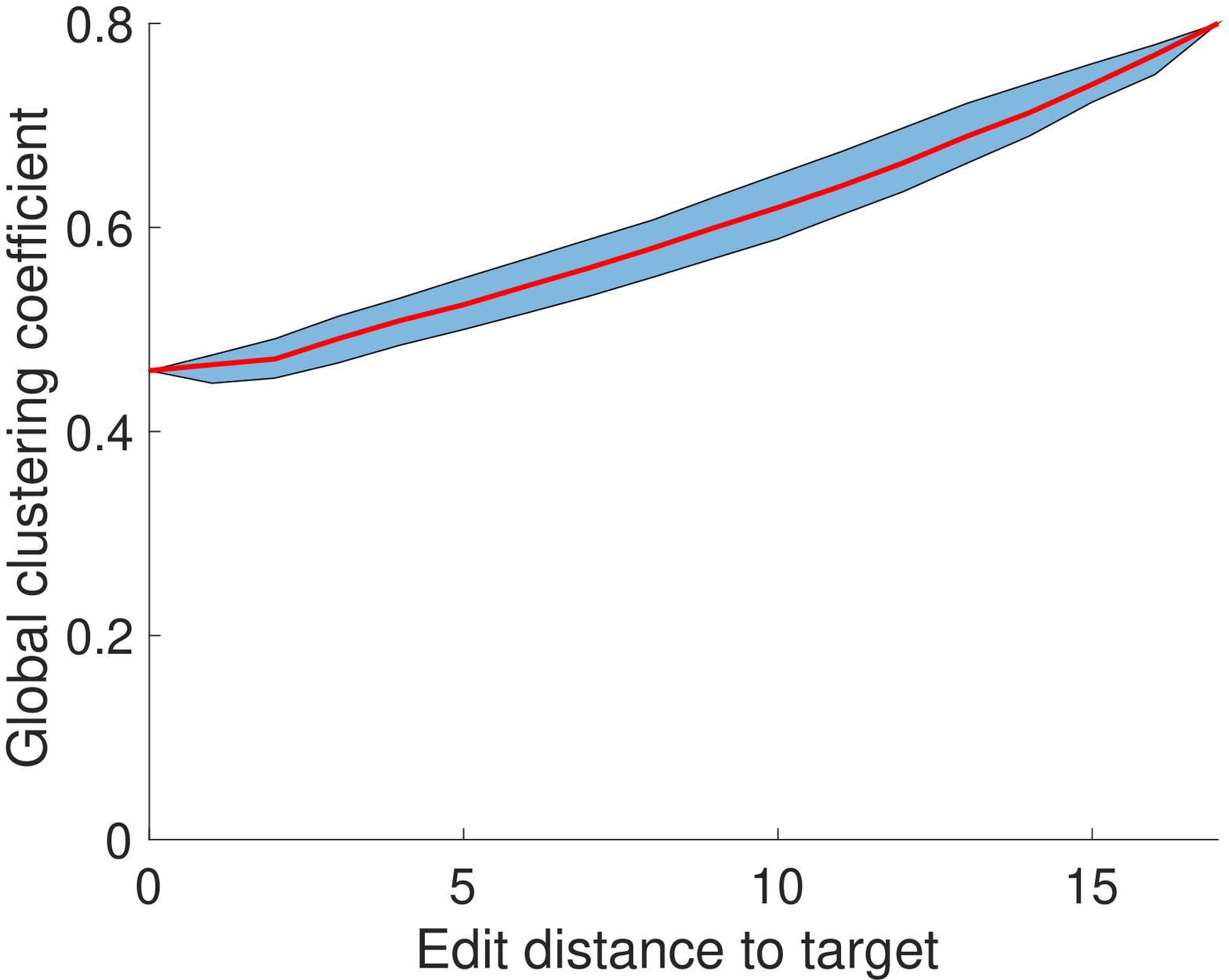}
		\caption{Snapshot 4 to 5.}
	\end{subfigure}
	\caption{Quartile box plots of global clustering coefficient vs.\ edit distance of the interpolations between snapshots of university freshmen friendships with rate parameter $s=1$. The quartiles are based on 10000 trials. The mean is in red and the lower and upper limits of the shaded region are the first and third quartiles, respectively.}
	\label{fig:vdbuntenvelopegc}
\end{figure}

\begin{figure}
	\begin{subfigure}{0.495\columnwidth}
		\includegraphics[width=0.95\linewidth]{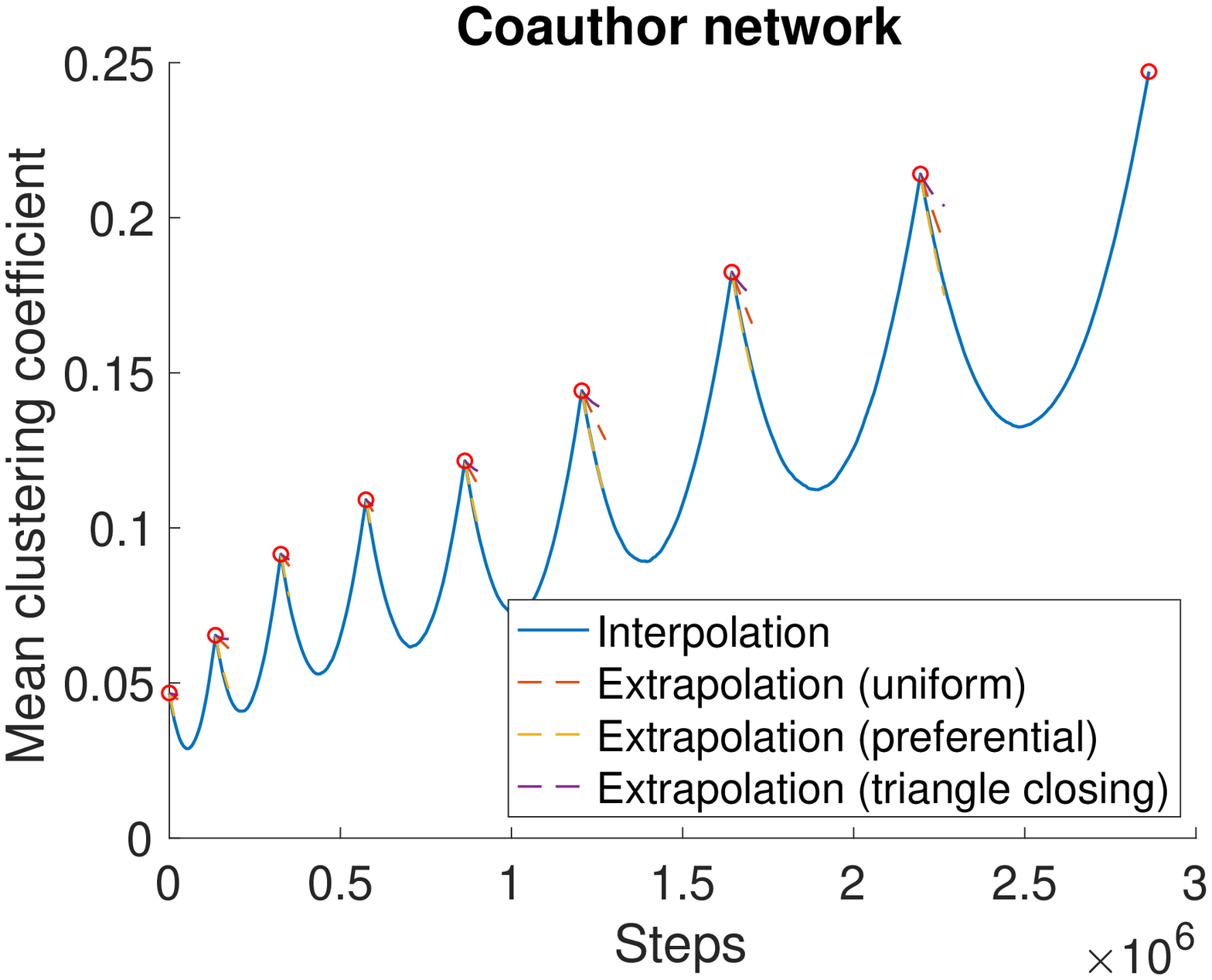}
	\end{subfigure}%
	\begin{subfigure}{0.495\columnwidth}
		\includegraphics[width=0.95\linewidth]{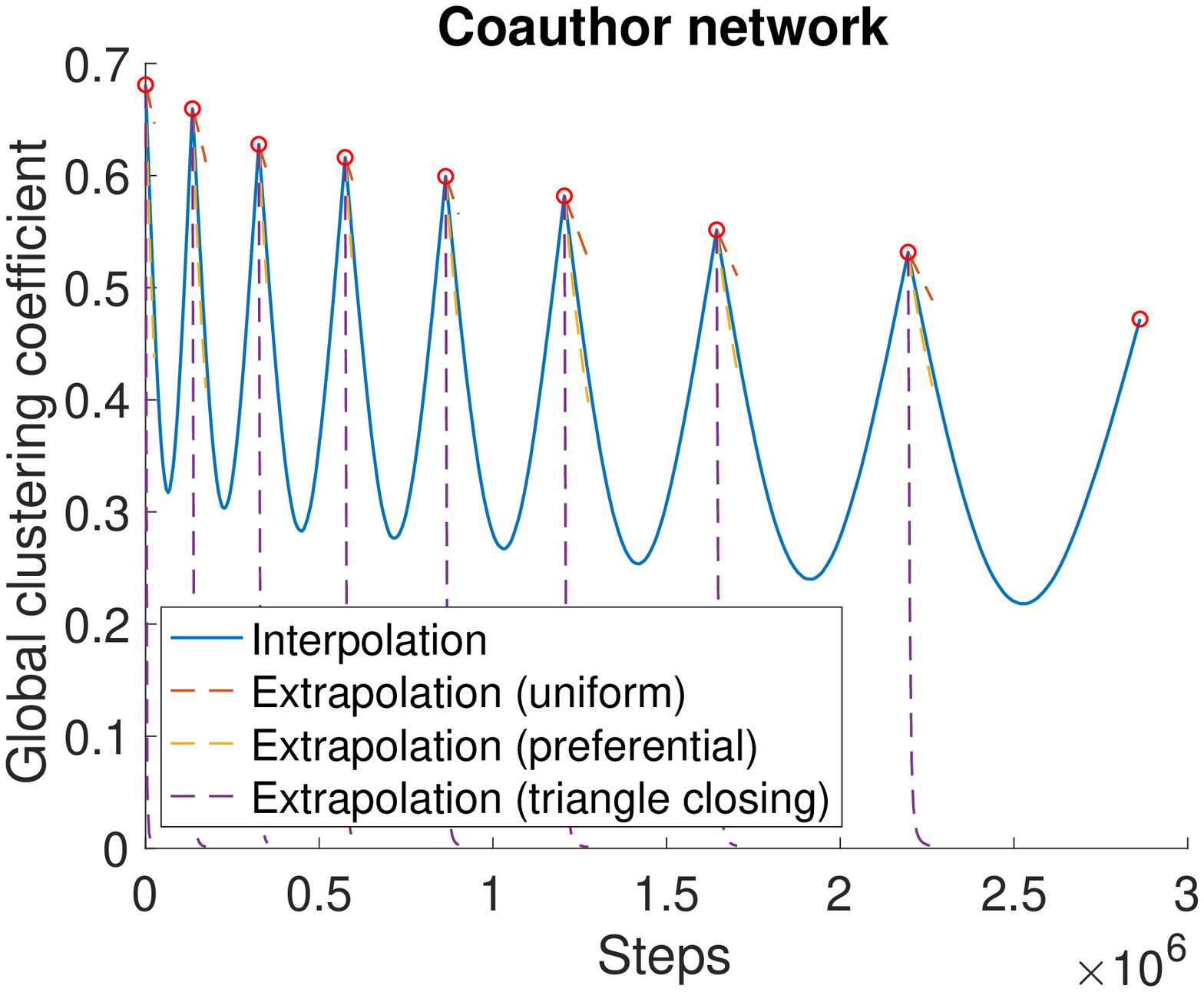}
	\end{subfigure}
	\caption{Evolution of mean and global clustering coefficient in interpolation between 9 snapshots of coauthor relationships. The rate parameter $s$ is $1$.}
	\label{fig:coauth}
\end{figure}

The interpolation and extrapolations in \Cref{fig:vdbunt} are shown for
only one run, but the qualitative behavior across runs is similar. Notably, we observe that our interpolation scheme viably interpolates the snapshots while the growth models fail to do so. For the
purposes of verifying consistency across different instances of our interpolation scheme, \Cref{fig:vdbuntenvelopegc}
provides a quartile box plot of the global clustering coefficient
as a function of edit distance for two pairs of sequential snapshots: one where the
beginning and ending clustering coefficients are similar, and one where they are not. The mean clustering coefficient has a comparable level of variability.

We repeated the comparison of interpolation and extrapolation using the DBLP
coauthorship dataset \cite{Ley-2009-DBLP}, 
which is a series of 18 graphs detailing the coauthor relationships among the conference proceedings papers on DBLP each year from 2000--2017. To prevent the clustering coefficients from being unduly influenced by outliers, we ignored papers from the dataset with more than 10 authors, and we aggregated the coauthorship data every 2 years (thus producing 9 snapshots). The comparison is shown in \Cref{fig:coauth}.

In this case, the clustering coefficients for the interpolation dip
substantially instead of staying roughly linear from one snapshot to
another. This is because the edits in the interpolation are done uniformly at
random. This presumably does not reflect the dynamics of the actual coauthorship
dynamic network, which preserves clustering structure where coauthorship forms
cliques in the network. This coauthorship example
shows that uniform edits are not a perfect model for describing changes in
a network, and more correlated edit rules may be necessary to preserve such a
high clustering coefficient. We leave this for future work.



\section{Conclusions}
We developed a model for network interpolation, analyzed quantities of interest,
and motivated it using snapshots from experiments. In doing so, we demonstrated
substantial improvements over extrapolation methods in reconstructing important
network statistics of dynamic graphs. We also provided a conceptual understanding of
trajectories between graphs in terms of edits and edit distances and opened an
area of investigation into the rules that should govern such trajectories. \rev{Although we analyzed the effect of the rate parameter in our model, further work remains to be done on the effect of sigmoid functions other than the standard logistic function and more complex edge edit rules beyond uniformly random edits. For example, edge edit rules that favor well-connectedness may help better preserve certain graph quantities of interest. Such insights would enhance the adaptability of our model to real-world situations.}

We emphasize that our model has broad scope: it can be applied to any set of
snapshots---real or synthetic---to generate random, plausible sequences of graphs
to move from a starting graph to a target graph in a wide variety of domains.
This makes our model applicable to
generating synthetic streaming datasets with planted structure (including
planted partition, planted clique, and planted coloring) as well as other types
of emerging structure in real-world networks. It also has the potential
to be used for on-demand or ``live'' data mining platforms.
Furthermore, our model is amenable to theoretical analysis and opens up
new theoretical directions in temporal network analysis.



\bibliographystyle{siamplain}
\bibliography{references}

\begin{thebibliography}{10}

\bibitem{Albert-1999-scaling}
{\sc R.~Albert and A.-L. Barab{\'a}si}, {\em {Emergence of scaling in random
  networks}}, Science, 286 (1999), pp.~509--512.

\bibitem{Aldecoa13}
{\sc R.~Aldecoa and I.~Mar\'{i}n}, {\em Exploring the limits of community
  detection strategies in complex networks}, Scientific Reports, 3 (2018).

\bibitem{Araujo-2014-Com2}
{\sc M.~Araujo, S.~Papadimitriou, S.~G\"{u}nnemann, C.~Faloutsos, P.~Basu,
  A.~Swami, E.~E. Papalexakis, and D.~Koutra}, {\em Com2: Fast automatic
  discovery of temporal (`comet') communities}, in Advances in Knowledge
  Discovery and Data Mining, Springer International Publishing, 2014,
  pp.~271--283, \url{https://doi.org/10.1007/978-3-319-06605-9_23}.

\bibitem{Backstrom-2006-group}
{\sc L.~Backstrom, D.~Huttenlocher, J.~Kleinberg, and X.~Lan}, {\em Group
  formation in large social networks: membership, growth, and evolution}, in
  Proceedings of the 12th {ACM} {SIGKDD} international conference on Knowledge
  discovery and data mining, ACM, 2006, pp.~44--54.

\bibitem{Ball-2007-DAWN}
{\sc J.~K. Ball}, {\em Drug abuse warning network, 2006: National estimates of
  drug-related emergency department visits}, tech. report, Substance Abuse and
  Mental Health Services Administration, Office of Applied Studies, 2007.

\bibitem{Barabasi99}
{\sc A.-L. Barab\'{a}si and R.~Albert}, {\em Emergence of scaling in random
  networks}, Science, 286 (1999), pp.~509--12.

\bibitem{Benson-2018-simplicial}
{\sc A.~R. Benson, R.~Abebe, M.~T. Schaub, A.~Jadbabaie, and J.~Kleinberg},
  {\em Simplicial closure and higher-order link prediction}, Proceedings of the
  National Academy of Sciences, 115 (2018), pp.~E11221--E11230.

\bibitem{Bhattacharjee18}
{\sc M.~Bhattacharjee, M.~Banerjee, and G.~Michailidis}, {\em Change point
  estimation in a dynamic stochastic block model},  (2018).

\bibitem{VanDeBunt-1999-friendship}
{\sc G.~G. V.~D. Bunt, M.~A.~V. Duijn, and T.~A. Snijders}, {\em Friendship
  networks through time: An actor-oriented dynamic statistical network model},
  Computational {\&} Mathematical Organization Theory,  (1999).

\bibitem{Damle18}
{\sc A.~Damle, V.~Minden, and L.~Ying}, {\em Simple, direct and efficient
  multi-way spectral clustering}, Information and Inference: A Journal of the
  IMA,  (2018), \url{https://doi.org/10.1093/imaiai/iay008}.

\bibitem{Dunlavy-2011-temporal}
{\sc D.~M. Dunlavy, T.~G. Kolda, and E.~Acar}, {\em Temporal link prediction
  using matrix and tensor factorizations}, {ACM} Transactions on Knowledge
  Discovery from Data, 5 (2011), pp.~1--27,
  \url{https://doi.org/10.1145/1921632.1921636}.

\bibitem{Farajtabar-2015-coevolve}
{\sc M.~Farajtabar, Y.~Wang, M.~G. Rodriguez, S.~Li, H.~Zha, and L.~Song}, {\em
  Coevolve: A joint point process model for information diffusion and network
  co-evolution}, in Advances in Neural Information Processing Systems, 2015,
  pp.~1954--1962.

\bibitem{Fowler-2008-happiness}
{\sc J.~H. Fowler and N.~A. Christakis}, {\em Dynamic spread of happiness in a
  large social network: longitudinal analysis over 20 years in the framingham
  heart study}, {BMJ}, 337 (2008), pp.~a2338--a2338.

\bibitem{Freeman-1980-semi}
{\sc L.~C. Freeman and S.~C. Freeman}, {\em A semi-visible college: Structural
  effects on a social networks group}, in Electronic communication: Technology
  and impacts, vol.~52, AAAS Washington, DC, 1980.

\bibitem{Ghasemian16}
{\sc A.~Ghasemian, P.~Zhang, A.~Clauset, C.~Moore, and L.~Peel}, {\em
  Detectability thresholds and optimal algorithms for community structure in
  dynamic networks}, Physical Review X, 6 (2016), p.~031005.

\bibitem{GVL}
{\sc G.~H. Golub and C.~F. Van~{L}oan}, {\em Matrix Computations (3rd Ed.)},
  Johns Hopkins University Press, Baltimore, MD, USA, 1996.

\bibitem{Gorovits-2018-LARC}
{\sc A.~Gorovits, E.~Gujral, E.~E. Papalexakis, and P.~Bogdanov}, {\em {LARC}:
  Learning activity-regularized overlapping communities across time}, in
  Proceedings of the 24th {ACM} {SIGKDD} International Conference on Knowledge
  Discovery, {ACM} Press, 2018, \url{https://doi.org/10.1145/3219819.3220118}.

\bibitem{Han-2004-evidence}
{\sc J.-D.~J. Han, N.~Bertin, T.~Hao, D.~S. Goldberg, G.~F. Berriz, L.~V.
  Zhang, D.~Dupuy, A.~J.~M. Walhout, M.~E. Cusick, F.~P. Roth, and M.~Vidal},
  {\em Evidence for dynamically organized modularity in the yeast
  protein--protein interaction network}, Nature, 430 (2004), pp.~88--93.

\bibitem{Henderson-2010-forensics}
{\sc K.~Henderson, T.~Eliassi-Rad, C.~Faloutsos, L.~Akoglu, L.~Li,
  K.~Maruhashi, B.~A. Prakash, and H.~Tong}, {\em Metric forensics: a
  multi-level approach for mining volatile graphs}, in Proceedings of the ACM
  SIGKDD international conference on Knowledge discovery and data mining, 2010.

\bibitem{sbm}
{\sc P.~W. Holland, K.~B. Laskey, and S.~Leinhardt}, {\em Stochastic
  blockmodels: First steps}, Social Networks, 5 (1983), pp.~109 -- 137.

\bibitem{Holme-2012-temporal}
{\sc P.~Holme and J.~Saram{\"a}ki}, {\em Temporal networks}, Physics reports,
  (2012).

\bibitem{Jackson07}
{\sc M.~O. Jackson and B.~W. Rogers}, {\em Meeting strangers and friends of
  friends: how random are social networks?}, The American Economic Review,
  (2007).

\bibitem{Jin01}
{\sc E.~M. Jin, M.~Girvan, and M.~E.~J. Newman}, {\em Structure of growing
  social networks}, Physical Review E, 64 (2001), p.~046132.

\bibitem{Kivela-2014-multilayer}
{\sc M.~Kivela, A.~Arenas, M.~Barthelemy, J.~P. Gleeson, Y.~Moreno, and M.~A.
  Porter}, {\em Multilayer networks}, Journal of Complex Networks, 2 (2014),
  pp.~203--271.

\bibitem{Komurov-2007-revealing}
{\sc K.~Komurov and M.~White}, {\em Revealing static and dynamic modular
  architecture of the eukaryotic protein interaction network}, Molecular
  Systems Biology, 3 (2007), \url{https://doi.org/10.1038/msb4100149}.

\bibitem{Kondor-2014-bitcoin}
{\sc D.~Kondor, I.~Csabai, J.~Sz\"{u}le, M.~P{\'{o}}sfai, and G.~Vattay}, {\em
  Inferring the interplay between network structure and market effects in
  bitcoin}, New Journal of Physics, 16 (2014), p.~125003.

\bibitem{Kossinets-2008-vector-clocks}
{\sc G.~Kossinets, J.~Kleinberg, and D.~Watts}, {\em The structure of
  information pathways in a social communication network}, in Proceeding of the
  {ACM} {SIGKDD} international conference on Knowledge discovery and data
  mining, 2008.

\bibitem{Krapivsky-2000-connectivity}
{\sc P.~L. Krapivsky, S.~Redner, and F.~Leyvraz}, {\em Connectivity of growing
  random networks}, Physical Review Letters, 85 (2000), pp.~4629--4632.

\bibitem{Leskovec-2008-planetary}
{\sc J.~Leskovec and E.~Horvitz}, {\em Planetary-scale views on a large
  instant-messaging network}, in Proceeding of the 17th international
  conference on World Wide Web, {ACM} Press, 2008,
  \url{https://doi.org/10.1145/1367497.1367620}.

\bibitem{Leskovec-2005-graphs}
{\sc J.~Leskovec, J.~Kleinberg, and C.~Faloutsos}, {\em Graphs over time:
  Densification laws, shrinking diameters and possible explanations}, in
  Proceeding of the eleventh {ACM} {SIGKDD} international conference on
  Knowledge discovery in data mining, 2005.

\bibitem{Lewis-2008-tastes}
{\sc K.~Lewis, J.~Kaufman, M.~Gonzalez, A.~Wimmer, and N.~Christakis}, {\em
  Tastes, ties, and time: A new social network dataset using facebook.com},
  Social Networks, 30 (2008), pp.~330--342,
  \url{https://doi.org/10.1016/j.socnet.2008.07.002}.

\bibitem{Ley-2009-DBLP}
{\sc M.~Ley}, {\em {DBLP}}, Proceedings of the {VLDB} Endowment, 2 (2009),
  pp.~1493--1500, \url{https://doi.org/10.14778/1687553.1687577}.

\bibitem{liu2018sampling}
{\sc P.~Liu, A.~Benson, and M.~Charikar}, {\em A sampling framework for
  counting temporal motifs}, arXiv:1810.00980,  (2018).

\bibitem{Michell-1997-pecking}
{\sc L.~Michell and A.~Amos}, {\em Girls, pecking order and smoking}, Social
  Science {\&} Medicine, 44 (1997), pp.~1861--1869.

\bibitem{Mislove-2008-growth}
{\sc A.~Mislove, H.~S. Koppula, K.~P. Gummadi, P.~Druschel, and
  B.~Bhattacharjee}, {\em Growth of the flickr social network}, in Proceedings
  of the first workshop on online social networks, {ACM} Press, 2008.

\bibitem{Mucha-2010-multislice}
{\sc P.~J. Mucha, T.~Richardson, K.~Macon, M.~A. Porter, and J.-P. Onnela},
  {\em Community structure in time-dependent, multiscale, and multiplex
  networks}, Science,  (2010), pp.~876--878,
  \url{https://doi.org/10.1126/science.1184819}.

\bibitem{Navlakha11}
{\sc S.~Navlakha and C.~Kingsford}, {\em Network archaeology: uncovering
  ancient networks from present-day interactions}, PLoS Computational Biology,
  7 (2011), p.~e1001119, \url{https://doi.org/10.1371/journal.pcbi.1001119}.

\bibitem{Overgoor-2018-choosing}
{\sc J.~Overgoor, A.~R. Benson, and J.~Ugander}, {\em Choosing to grow a graph:
  Modeling network formation as discrete choice}, arXiv:1811.05008,  (2018).

\bibitem{Panzarasa-2009-CollegeMsg}
{\sc P.~Panzarasa, T.~Opsahl, and K.~M. Carley}, {\em Patterns and dynamics of
  users' behavior and interaction: Network analysis of an online community},
  Journal of the American Society for Information Science and Technology, 60
  (2009), pp.~911--932.

\bibitem{Paranjape-2017-motifs}
{\sc A.~Paranjape, A.~R. Benson, and J.~Leskovec}, {\em Motifs in temporal
  networks}, in Proceedings of the Tenth ACM International Conference on Web
  Search and Data Mining, ACM, 2017, pp.~601--610.

\bibitem{Scholtes-2017-networks}
{\sc I.~Scholtes}, {\em When is a network a network?: Multi-order graphical
  model selection in pathways and temporal networks}, in Proceedings of the ACM
  SIGKDD international conference on Knowledge discovery and data mining, 2017.

\bibitem{KS}
{\sc J.~L. Snell and J.~G. Kemeny}, {\em Finite Markov Chains},
  Springer--Verlag, New York, NY, USA, 1976.

\bibitem{Snijders-1996-SAOM}
{\sc T.~A. Snijders}, {\em Stochastic actor-oriented models for network
  change}, The Journal of Mathematical Sociology, 21 (1996), pp.~149--172.

\bibitem{Xu14}
{\sc K.~S. Xu and A.~O.~H. III}, {\em Dynamic stochastic blockmodels for
  time-evolving social networks}, IEEE Journal of Selected Topics in Signal
  Processing, 8 (2014), pp.~552--562.

\bibitem{Young2018}
{\sc J.-G. Young, L.~H{\'e}bert-Dufresne, E.~Laurence, C.~Murphy, G.~St-Onge,
  and P.~Desrosiers}, {\em Network archaeology: phase transition in the
  recoverability of network history}, arXiv:1803.09191,  (2018).

\end{thebibliography}

\end{document}